\documentclass[a4paper,10pt]{article}

\usepackage[british]{babel}

\usepackage{amsmath, amssymb, amsfonts, amsthm}
\usepackage{authblk}
\usepackage{wrapfig}

\usepackage{enumerate}  
\usepackage{dsfont}    
\usepackage[colorlinks,
			pdffitwindow=false,
      plainpages=false,
      pdfpagelabels=true,
     	pdfpagemode=UseOutlines,
      pdfpagelayout=SinglePage,
			colorlinks=true,
      hyperfootnotes=false,
			linkcolor=blue,
			citecolor=green!50!black]{hyperref}

\usepackage{color}     
\usepackage{tikz}       
\usepackage{nccmath}
\usepackage{amsmath}   

\usepackage[leftmargin=3em]{quoting}  

\usepackage{esint}  

\makeatletter
  \newcommand{\eqnum}{\leavevmode\hfill\refstepcounter{equation}\textup{\tagform@{\theequation}}} 
\makeatother

\newcommand{\sfA}{\ensuremath{\mathsf{A}}}
\newcommand{\sfB}{\ensuremath{\mathsf{B}}}
\newcommand{\A}{\ensuremath{\mathcal{A}}}

\newcommand{\bk}{{\bar{k}}}
\renewcommand{\d}{\mathrm{d}}

\newcommand{\e}{\mathrm{e}}

\newcommand{\E}{\mathcal{E}}

\newcommand{\bw}{\mathrm{bw}}
\newcommand{\fw}{\mathrm{fw}}

\newcommand{\F}{\mathcal{F}}

\newcommand{\fast}{\mathrm{fast}}
\newcommand{\slow}{\mathrm{slow}}

\renewcommand{\H}{\mathcal{H}}

\renewcommand{\L}{\mathcal{L}}
\newcommand{\M}{\mathcal{M}}
\mathchardef\mhyphen="2D 

\renewcommand{\P}{\mathcal{P}}

\newcommand{\Q}{\mathcal{Q}}

\newcommand{\R}{\mathcal{R}}
\newcommand{\RR}{\mathbb{R}}

\newcommand{\tot}{\mathrm{tot}}
\newcommand{\tr}{\mathrm{tr}}
\newcommand{\re}{\mathrm{re}}
\newcommand{\bW}{{\bar{W}}}

\newcommand{\W}{\mathcal{W}}
\newcommand{\Z}{\mathbb{Z}}
\newcommand{\X}{\mathcal{X}}
\newcommand{\Y}{\mathcal{Y}}

\newcommand{\tp}{^{\mathsf{T}}}

\let\div\relax
\DeclareMathOperator\div{div}

\DeclareMathOperator\grad{\nabla\!}
\DeclareMathOperator\Grad{grad}

\DeclareMathOperator\Prob{Prob}

\DeclareMathOperator\Ran{Ran}

\DeclareMathOperator\lapl{\Delta}

\newcommand{\super}[1]{^{\scriptscriptstyle{(#1)}}}

\newtheorem{theorem}{Theorem}[section]

\newtheorem{lemma}[theorem]{Lemma}
\newtheorem{proposition}[theorem]{Proposition}

\newtheorem{definition}[theorem]{Definition}
\newenvironment{remark}
  {\par\medbreak\refstepcounter{theorem}%
    \noindent\textbf{Remark~\thetheorem. }}%
  {\qed\par\medskip}

\title{Gradient and Generic systems in the space of fluxes, applied to reacting particle systems}
\author{D.R. Michiel Renger}
\affil{\small{WIAS Berlin\\Mohrenstrasse 39\\10117 Berlin, Germany\\tel: +49 30 20372-470\\fax: +49 30 20372-404}}

\date{\today}

\begin{document}

\maketitle




\begin{abstract}
In a previous work we devised a framework to derive generalised gradient systems for an evolution equation from the large deviations of an underlying microscopic system, in the spirit of the Onsager-Machlup relations. Of particular interest is the case where the microscopic system consists of random particles, and the macroscopic quantity is the empirical measure or concentration. In this work we take the particle flux as the macroscopic quantity, which is related to the concentration via a continuity equation. By a similar argument the large deviations can induce a generalised gradient or Generic system in the space of fluxes. In a general setting we study how flux gradient or generic systems are related to gradient systems of concentrations. The arguments are explained by the example of reacting particle systems, which is later expanded to include spatial diffusion as well.
\end{abstract}

\section{Introduction}

By the Boltzmann-Einstein relation, the free energy of a system is inherently related to the fluctuations of an underlying microscopic particle system. In this sense, two systems with the same macroscopic behaviour can be driven by completely different free energies, if their corresponding microscopic systems are different. Therefore, one of the main objectives of (equilibrium) statistical mechanics is to derive the `physically correct' macroscopic free energy from fluctuations in microscopic systems. A similar principle can be applied to systems that evolve over time, where dynamic fluctuations may lead to a gradient flow, driven by the free energy. For stochastically reversible systems and close to equilibrium this is the classic Onsager-Machlup theory~\cite{Onsager1931I,Onsager1953I}. Such relations are known to hold for many reversible dynamics, not necessarily close to equilibrium~\cite{Bertini2004}. More recently, it was shown that microscopic reversibility always implies the emergence of a macroscopic gradient flow~\cite{MielkePeletierRenger2014}, but in this generality one needs to allow for so-called generalised gradient flows. In brief, a generalised gradient structure (GGS) is defined by a possibly non-linear relation between velocities and affinities. Although there are exist some non-reversible models that lead to macroscopic gradient flows~\cite{Dietert2015}, these are considered non-typical; so in order to understand systems with non-reversible microscopic fluctuations, one needs to look for even further (thermodynamically consistent) generalisations of a gradient flow.

One such generalisation could be Generic~\cite{Ottinger2005}, which couples a gradient flow to a Hamiltonian system in such a way that the Hamiltonian energy is conserved and the free energy decays over time. Clearly, in order to derive Generic structures from dynamical large deviations in a general setting, one again needs to allow for non-linear relations, just like generalised gradient flows. One is thus lead to study generalised Generic structures (GGEN)~\cite{Mielke2011GGEN}. A recent work shows a set of necessary and sufficient conditions of which microscopic fluctuations induce such a generalised Generic structure~\cite{KLMP2017phys,KLMP2018math}.

Another generalisation worth considering is to study, in addition to the macroscopic state variables, the corresponding fluxes. These fluxes hold more information than the state variables due to the possible occurrence of\, `divergence-free' fluxes that do not alter the states. The study of flux large deviations and corresponding thermodynamic properties is a key concept in Macroscopic Fluctuation Theory~\cite{Bertini2015MFT}. With the fluctuations on the fluxes one could for example extract a generalisation of a gradient structure: (what we call) a force structure \cite{MaesNetocny2008,KaiserJackZimmer2017,Maes2017,Renger2017}, where the affinity/driving force may no longer be the gradient of some free energy.

In this paper we pursue another flux-based point of view, based on the observation that a GGS and GGEN can be interpreted as a free energy balance. If there would be work done that results in a divergence-free flux, than one might expect a gap the energy balance, so that the such systems can not induce a GGS or GGEN. Could one then still have GGS/GGEN structure in the space of fluxes? The main point of this paper is that this is generally impossible. We will see that if the fluctuations induce a GGS or GGEN in the space of fluxes, this is, up to some conditions, equivalent to the fluctuations inducing a GGS or GGEN in the state space.

Our leading example and main application will be that of an isothermal chemical reaction network, as studied in~\cite{MielkePattersonPeletierRenger2017}. In Section~\ref{sec:chemical reactions} we recall the main arguments from that paper, applied to concentrations undergoing reactions. In Section~\ref{sec:reaction fluxes} we expand these ideas and show how they can be applied to reaction fluxes. Based on this example, we then develop an abstract theory about induced GGSs and GGENs in flux space, in Section~\ref{sec:general theory}. To show the generality of these principles, we then show in Sections~\ref{sec:diffusion} how the theory applies to transport fluxes in a diffusing particles model, and in Section~\ref{sec:reaction-diffusion} we combine the arguments from Sections~\ref{sec:reaction fluxes} and \ref{sec:reaction-diffusion} to derive results for a transport and reaction fluxes in a simple reaction-diffusion model.

\section{Leading example 1: chemical reactions}
\label{sec:chemical reactions}

In this section we explain the main concepts, mostly by reiterating the arguments of \cite{MielkePeletierRenger2014} and \cite{MielkePattersonPeletierRenger2017}. In particular, we argue that large deviations/fluctuations provide a `physically correct' GGS for the evolution of concentrations undergoing chemical reactions. 

Consider a network of isothermal chemical reactions, for example:
\begin{align*}
  2\mathsf{Na} + \mathsf{Cl}_2 \to 2\mathsf{NaCl} &&\text{and}&&  2\mathsf{NaCl} \to 2\mathsf{Na} + \mathsf{Cl}_2.
\end{align*}
We denote the set of species by $\Y$ (in this example $\{\mathsf{Na},\mathsf{Cl}_2,\mathsf{NaCl}\}$) and the set of reactions by $\R$; for each reaction $r$ we consider a forward and backward reaction (so that here $\R$ consists of one element). The stoichiometric coefficients are denoted by $(\alpha_{r,y})_{y\in\Y}$ for the forward reactants (here $(2,1,0)$) and $(\beta_{r,y})_{y\in\Y}$ for the forward products (here $(0,0,2)$), which yields the state change matrix $\Gamma:=(\gamma_{r,y}):=(\beta_{r,y}-\alpha_{r,y})$. The evolution of the concentrations $\rho_t\in\RR^\Y$ is then described by the \emph{Reaction Rate Equation},
\begin{equation}
  \dot\rho_t = \sum_{r\in\R}\gamma_r \big(k_{r,\fw}(\rho_t)-k_{r,\bw}(\rho_t)\big)=\Gamma \bk(\rho_t),
\label{eq:RRE1}
\end{equation}
with concentration-dependent reaction rates $\bk_r:=k_{r,\fw} - k_{r,\bw}$.

Typically (and often in this paper), the reaction rates will be of the form
\begin{align}
  k_{r,\fw}(\rho):=\kappa_{r,\fw} \rho^{\alpha_r}
  &&\text{and}&&
  k_{r,\bw}(\rho):=\kappa_{r,\bw} \rho^{\beta_r}
\label{eq:RRE1 mass-action kinetics}
\end{align}
for some constants $\kappa_{r,\fw},\kappa_{r,\bw}$, using the notation $\rho^{\alpha_r}=\prod_{y\in\Y} \rho_y^{\alpha_{r,y}}$. In that case the network is said to be of \emph{mass-action kinetics}. For more background on chemical reaction networks we refer to the survey~\cite{AndersonKurtz2011}. More details about the fluctuations and induced GGSs for chemical reaction networks can be found in \cite{MielkePattersonPeletierRenger2017}; for completeness we shall recall these results in this section.

To notationally stress the similarity and differences between different concepts throughout this paper, we shall always net quantities by a bar $(\bar{\quad}$), and we distinguish functionals on state space from functionals on flux space by a hat ($\hat\quad$). In this section we study concentrations only, which we consider to be states.

\subsection{Reacting particle system}

A classical microscopic particle system underlying the evolution~\eqref{eq:RRE1} is the following~\cite{Kurtz1972}. Let $V$ be a large, well-mixed volume that contains, at time $t$, a total number $N\super{V}_t$ of particles of species $Y_{t,i}\super{V}, i=1,\hdots,N\super{V}_t$. A reaction $r$ occurs randomly with some propensities (jump rates) $\lambda\super{V}_{r,\fw}(\rho), \lambda\super{V}_{r,\bw}(\rho)$. Whenever a forward reaction $r$ occurs, $\alpha_r$ particles are removed and $\beta_r$ particles are created, and vice versa for a backward reaction. Hence each reaction requires a cumbersome relabelling of particles $Y_{t,i}$. It is therefore more practical to work directly with the (particles per volume) concentration $\rho_{t,y}\super{V}:=\tfrac1V\#\{Y_{t,i}=y, i=1,\hdots N_t\super{V}\}$. This quantity will also play the role of the macroscopic state variable. Whenever a forward or backward reaction $r$ takes place, the concentration can now be simply updated by a jump $\rho_t\super{V}=\rho\super{V}_{t^-} \pm \tfrac1V\gamma_r$. Then the $\Y$-dimensional vector $\rho\super{V}_t$ is a Markov jump process, which satisfies the master equation
\begin{align*}
  \dot {\hat{P}}_t\super{V}(\rho) = \sum_{r\in\R}
    &\lambda\super{V}_{r,\fw}(\rho-\tfrac1V\gamma_r)\hat P_t\super{V}(\rho-\tfrac1V\gamma_r) - \lambda\super{V}_{r,\fw}(\rho)\hat P_t\super{V}(\rho) +\\
    &\lambda\super{V}_{r,\bw}(\rho+\tfrac1V\gamma_r)\hat P_t\super{V}(\rho+\tfrac1V\gamma_r) - \lambda\super{V}_{r,\bw}(\rho)\hat P_t\super{V}(\rho).
\end{align*}
It will be beneficial to work with the corresponding generator, which is the adjoint of the right-hand side of the master equation, with respect to the dual pairing $\langle\hat P\super{V}_t,f\rangle=\sum_\rho f(\rho)\,\hat P\super{V}_t(\rho)$ with an arbitrary test function:
\begin{equation}
  (\hat \Q\super{V}f)(\rho) := \sum_{r\in\R} \lambda\super{V}_{r,\fw}(\rho)\big( f(\rho+\tfrac1V\gamma_r)-f(\rho) \big) + \lambda\super{V}_{r,\bw}(\rho)\big( f(\rho-\tfrac1V\gamma_r)-f(\rho) \big).
\label{eq:RRE1 generator}
\end{equation}

The propensities that are usually used in the so-called chemical master equation are derived from combinatoric considerations, and yield the mass-action kinetics in the limit~\cite{Kurtz1972,AndersonKurtz2011,MielkePattersonPeletierRenger2017}:
\begin{align}
  \tfrac1V \lambda\super{V}_{r,\fw}(\rho):=\frac1V\cdot\frac{\kappa_{r,\fw}}{V^{\alpha_{r,\tot}-1}} \frac{(\rho V)!}{(\rho V-\alpha_r)!} \xrightarrow{V\to\infty}\kappa_{r,\fw} \rho^{\alpha_r}=k_{r,\fw}(\rho),\notag\\
  \tfrac1V \lambda\super{V}_{r,\bw}(\rho):=\frac1V\cdot\frac{\kappa_{r,\bw}}{V^{\beta_{r,\tot}-1}} \frac{(\rho V)!}{(\rho V-\beta_r)!} \xrightarrow{V\to\infty}\kappa_{r,\bw} \rho^{\beta_r}=k_{r,\bw}(\rho),
\label{eq:RRE1 convergent reaction rates}
\end{align}
using the notation $\alpha_{r,\tot}:=\sum_{y\in\Y}\alpha_{r,y}$ and $\alpha_r!:=\prod_{y\in\Y}\alpha_{r,y}!$.

\subsection{Equilibrium: limit, large deviations and free energy}

We will from now on (throughout this section) assume that the reaction network is of mass-action kinetics~\eqref{eq:RRE1 mass-action kinetics}, and chemically detailed balanced, i.e. there exists a $\rho^*\in\RR_+^\Y$ for which 
\begin{equation}
  \kappa_{r,\fw}{\rho^*}^{\alpha_r}=\kappa_{r,\bw}{\rho^*}^{\beta_r} \qquad\text{for all } r\in\R.
\label{eq:RRE1 chem detailed balance}
\end{equation} 
Naturally, $\rho^*$ is an equilibrium under the deterministic evolution~\eqref{eq:RRE1} \footnote{It should be stressed that, given an initial concentration $\rho_0$, both the deterministic evolution and the stochastic model is confined to the `stoichiometric compatibility class' $\rho_0+\Ran\Gamma=\{\rho_0+\Gamma w: w\in\RR_+^\R\}$. Therefore, it is not clear whether this equilibrium $\rho^*$ lies within the compatibility class that corresponds to the initial concentration. However, if there exists a detailed balanced concentration $\rho^*$, then there exists a unique detailed balanced concentration within each such class, see~\cite{AndersonKurtz2011} and the references therein. Without loss of generality, we can therefore implicitly assume that the detailed balanced equilibrium is unique, and lies within the correct compatibility class.}. Under this assumption, the invariant distribution of the stochastic model is known to be \cite{AndersonKurtz2011}:
\begin{equation*}
  \hat P_\infty\super{V}(\rho) = \prod_{y\in\Y} \frac{(V\rho_y^*)^{V\rho_y}}{(V\rho_y)!}\e^{-V\rho_y^*}.
\end{equation*}
Letting $V\to\infty$, this invariant distribution concentrates on the equilibrium state:
\begin{equation*}
  \hat P_\infty\super{\infty}(\rho)=
    \begin{cases}
      1, & \rho=\rho^*,\\
      0, & \text{otherwise}.
    \end{cases}
\end{equation*}
One can then extract the free energy by considering the corresponding large deviations, i.e. the exponential rate which with $\hat P_\infty\super{V}$ converges to zero. Indeed, by Stirling's formula,
\begin{equation}
  -\tfrac1V \log \hat P_\infty\super{V}(\rho)\xrightarrow{V\to\infty} \sum_{y\in\Y} \rho_y\log\mfrac{\rho_y}{\rho_y^*} - \rho_y + \rho_y^* =: h(\rho|\rho^*).
\label{eq:RRE1 Stirling}
\end{equation}
Such limit is known as a large-deviation principle; the function on the right characterises the stochastic cost of microscopic fluctuations. If the reaction rates are related to an internal energy via Arrhenius' law, than the expression $h(\rho|\rho^*)$ is really the Helmholtz free energy, apart from a normalisation term and a constant scaling, as explained in more detail in~\cite[Sec.~2.2\&2.3]{MielkePattersonPeletierRenger2017}.

\subsection{Dynamics: limit and large deviations}
\label{subsec:RRE1 dynamics limit ldp}

Observe that in the microscopic model, the process speeds up as $V$ increases with order $V$ while the jump sizes are of size $1/V$. Therefore by \eqref{eq:RRE1 convergent reaction rates}, as $V\to\infty$ the generator $\hat \Q\super{V}$ converges to the limit generator
\begin{equation*}
  (\hat \Q\super{\infty}f)(\rho) := \sum_{r\in\R} \bk_r(\rho)\grad f(\rho)\cdot\gamma_r.
\end{equation*}
Since this generator depends on the test function $f$ through $\grad f(\rho)$ only, we can make the ansatz that the limit process is deterministic $\hat P\super{\infty}_t(\tilde \rho)=\delta_{\rho_t}(\tilde \rho)$, for some curve $\rho_{(\cdot)}$. Plugging this into the definition of the generator yields:
\begin{multline*}
  \grad f(\rho_t)\cdot\dot\rho_t=\partial_t f(\rho_t)=\partial_t \langle\hat P\super{\infty}_t,f\rangle = \langle\hat P_t\super{\infty},\Q\super{\infty} f\rangle \\
 = (\hat \Q\super{\infty} f)(\rho_t) = \grad f(\rho_t)\cdot\sum_{r\in\R} \bk_r(\rho_t)\gamma_r.
\end{multline*}
As this relation holds for any test function~$f$, we see that the ansatz was justified if the postulated curve $\rho_t$ satisfies the Reaction Rate Equation~\eqref{eq:RRE1}. Hence the stochastic process $\rho\super{V}_t$ converges (pathwise in probability) to the deterministic solution $\rho_t$ of the Reaction Rate Equation.

Similar to the calculation of the fluctuations of the equilibrium~\eqref{eq:RRE1 Stirling}, we can study the large deviations of the path probabilities $\hat P\super{V}$; this is known as a dynamic large-deviation principle. These dynamical fluctuations can be formally calculated with the framework of~\cite{Feng2006}. To this aim we study the non-linear generator:
\begin{align*}
  (\hat \H\super{V}f)(\rho)
    &:=\tfrac1V \e^{-Vf(\rho)}(\hat \Q\super{V}\e^{Vf})(\rho) \\
    &=\sum_{r\in\R} \tfrac1V \lambda\super{V}_{r,\fw}(\rho)\big(\e^{Vf(\rho+\tfrac1V\gamma_r)-Vf(\rho)}-1\big) \\
    &\hspace{0.5cm} + \tfrac1V \lambda\super{V}_{r,\bw}(\rho)\big(\e^{Vf(\rho-\tfrac1V\gamma_r)-Vf(\rho)}-1\big) \\
    &\xrightarrow{V\to\infty} \sum_{r\in\R} k_{r,\fw}(\rho) \big(\e^{\grad f(\rho)\cdot\gamma_r}-1\big) + k_{r,\bw}(\rho) \big(\e^{-\grad f(\rho)\cdot\gamma_r}-1\big).
\end{align*}
As before, the limit depends on the test function through $\grad f(\rho)$ only, which is consistent with the fact that the limit is deterministic. We then define, by a slight abuse of notation,
\begin{align}
  \hat\H(\rho,\xi)&:=\sum_{r\in\R}k_{r,\fw}(\rho) \big(\e^{\xi\cdot\gamma_r}-1\big) + k_{r,\bw}(\rho) \big(\e^{-\xi\cdot\gamma_r}-1\big),
  &\text{and}
\label{eq:RRE1 H}\\
  \hat\L(\rho,s) &:= \sup_{\xi\in\RR^\Y} \xi\cdot s - \hat\H(\rho,\xi).
\label{eq:RRE1 L}
\end{align} 
The dynamic large-deviation principle now states that
\footnote{The rigorous definition of the large-deviation principle, the heuristics behind this method, and the rigorous proof of this statement is all beyond the scope of this paper. For the precise details we refer to \cite{Feng2006}, and for the rigorous proof for this particular system (by more classical methods) to~\cite{DupuisRamananWu2016,AgazziDemboEckmann2017a,PattersonRenger2018flux}. For the sake of brevity, we assume that the randomness in the initial condition is sufficiently small (e.g. deterministic) so that we do not obtain initial fluctuations.}
\begin{equation}
  \Prob\super{V}\!\big(\rho\super{V}_{(\cdot)}\approx \rho_{(\cdot)}\big) \stackrel{V\to\infty}{\sim} \e^{-V\int_0^T\!\hat \L(\rho_t,\dot\rho_t)\,dt}.
\label{eq:RRE1 dynamic ldp}
\end{equation}

\begin{remark} The function~\eqref{eq:RRE1 L} is implicitly defined as a supremum; although the supremum can be calculated explicitly, this leads to very cumbersome expressions. However, it does have a dual formulation in terms of a minimisation problem:
\begin{equation*}
  \hat\L(\rho,s) = \inf_{\substack{j_\fw,j_\bw\in\RR^\R_+:\\\Gamma (j_\fw-j_\bw)=s}}\, h\big(j_\fw | k_\fw(\rho)\big) + h\big(j_\bw | k_\bw(\rho)\big),
\end{equation*}
where $h(j| k):=\sum_{r\in\R} j_r\log(j_r/k_r)-j_r+k_r$, similar to \eqref{eq:RRE1 Stirling}. Although the relative entropy $h$ appears in both expressions, they should not be confused: $h(\rho|\rho^*)$ is an equilibrium rate whereas $h(j|k(\rho))$ is a dynamic quantity. We shall see later on that the latter can be directly (without the infimum) be interpreted as a large-deviation rate, where the variable $j$ is a reaction flux.
\label{rem:RRE1 entropic form}
\end{remark}

\subsection{GGS, energy balance, and relation with fluctuations}
\label{subsec:RRE1 GGS}

The (naive) aim is to rewrite the macroscopic equation as a gradient flow of some free energy $\F$:
\begin{equation}
  \dot\rho_t=-\hat K(\rho_t)\grad\hat \F(\rho_t)=:-\Grad_{\rho_t}\hat \F(\rho_t),
\label{eq:linear gradient flow}
\end{equation}
where $\hat K(\rho)$ is some linear symmetric, positive definite (linear response) operator that maps thermodynamic forces to velocities. Mathematically, this operator can be interpreted as the inverse of the metric tensor of some manifold, so that the right-hand side is the gradient on this manifold.

Clearly, \eqref{eq:linear gradient flow} is equivalent to requiring
\begin{equation*}
  0 = \tfrac12\lVert \dot\rho_t + \hat K(\rho_t)\grad\hat \F(\rho_t)\rVert^2_{\hat K(\rho_t)^{-1}} =  \tfrac12\lVert \dot\rho_t\rVert^2_{\hat K(\rho_t)^{-1}} + \tfrac12\lVert \grad\hat \F(\rho_t)\rVert^2_{\hat K(\rho_t)} + \grad\hat \F(\rho_t)\cdot\dot\rho_t,
\end{equation*}
if we set $\lVert \xi \rVert_{\hat K(\rho)}^2:=\langle \xi,\hat K(\rho)\xi\rangle$ and $\lVert s\rVert_{\hat K(\rho)^{-1}}:=\langle s,\hat K(\rho)^{-1} s \rangle$. Integrated over a time interval $(0,T)$, this reads
\begin{equation*}
  0 = \int_0^T\!\Big(\tfrac12\lVert \dot\rho_t\rVert^2_{\hat K(\rho_t)^{-1}} + \tfrac12\lVert \grad\hat \F(\rho_t)\rVert^2_{\hat K(\rho_t)}\Big)\,dt + \hat \F(\rho_T)-\hat \F(\rho_0).
\end{equation*}
The last two terms describe the free energy loss (or entropy production), and the first two terms describe the dissipation\footnote{For a linear gradient flow \eqref{eq:linear gradient flow},  $\tfrac12\lVert \dot\rho_t\rVert^2_{\hat K(\rho_t)^{-1}} + \tfrac12\lVert \grad\hat \F(\rho_t)\rVert^2_{\hat K(\rho_t)}=\lVert \dot\rho_t\rVert^2_{\hat K(\rho)^{-1}}$; hence the dissipation can be seen as a kinetic energy. This is however no longer true for general GGSs.}; as such this equation represents a free energy balance. 

Observe that this expression is always non-negative, and $0$ exactly on the gradient flow~\eqref{eq:linear gradient flow}. Moreover, we see that this expression has the same dimension as $\hat \F$, the free energy; it is indeed the free energy cost to deviate from the macroscopic dynamics. This interpretation shows that this cost should be equal to the cost $\int_0^T\!\hat \L(\rho_t,\dot\rho_t)\,dt$ of microscopic fluctuations~\eqref{eq:RRE1 dynamic ldp}.

This is in many cases, and particularly in this case of chemical reactions, impossible. Since the large-deviation function \eqref{eq:RRE1 L} is non-quadratic, one should allow for non-quadratic dissipation terms. We therefore replace the two squared norms by a pair of dual~\emph{dissipation potentials}:
\begin{align}
  0 = \int_0^T\!\Big(\hat \Psi(\rho_t,\dot\rho_t) + \hat \Psi^*\big(\rho_t,-\grad\hat \F(\rho_t)\big)\Big)\,dt + \hat \F(\rho_T)-\hat \F(\rho_0),
\label{eq:nonlinear EDI}
\end{align}
where, as in the quadratic case, the potentials are convex duals of each other, i.e. $\hat \Psi(\rho,s)=\sup_{\xi} \xi\cdot s - \hat \Psi^*(\rho,\xi)$. Moreover, we assume that $\hat \Psi$ and $\hat \Psi^*$ are both non-negative; from~\eqref{eq:nonlinear EDI} we then see that the free energy $\F$ is non-increasing on the flow.

Due to the convex duality, the right-hand side of \eqref{eq:nonlinear EDI} is again always non-negative, and $0$ exactly when
\begin{equation*}
  \dot\rho_t = \grad_\xi\hat \Psi^*\big(\rho_t,-\grad\hat \F(\rho_t)\big).
\end{equation*}
We call such equation a~\emph{generalised gradient flow}, and the underlying structure $(\RR^\Y,\hat \Psi,\hat \F)$ a \emph{generalised gradient structure (GGS)}. Note that the generalisation with respect to \eqref{eq:linear gradient flow} lies in the fact that we allow for a non-linear relation between forces and velocities. We moreover say that a gradient structure is \emph{induced} by a cost function $\hat \L$ whenever $\hat \L(\rho,s) = \hat \Psi(\rho,s)+\hat \Psi^*\big(\rho,-\grad\hat \F(\rho)\big) + \grad\hat \F(\rho)\cdot s$.

In Section~\ref{sec:general theory} we will recall the relation between fluctuation costs and GGSs, as described in \cite{MielkePeletierRenger2014}. Applied to the current setting of chemical reactions, we reiterate the following result from~\cite{MielkePattersonPeletierRenger2017}. If we again assume mass-action kinetics~\eqref{eq:RRE1 mass-action kinetics} and chemical detailed balance~\eqref{eq:RRE1 chem detailed balance}, then there exists a unique GGS $(\RR^\Y,\hat \Psi,\hat \F)$ induced by the large-deviation cost $\hat \L$ from~\eqref{eq:RRE1 L}, where
\begin{align}
  \hat \F(\rho):=\tfrac12h(\rho | \rho^*),
  &&\text{and}&&
  \hat \Psi^*(\rho,\xi):=\sum_{r\in\R}\hat\sigma_r(\rho)\big(\cosh(\xi\cdot\gamma_r)-1\big),
\label{eq:RRE1 GGS}
\end{align}
with $\hat\sigma_r(\rho):=2\sqrt{k_{r,\fw}(\rho)k_{r,\bw}(\rho)}$. Since $\hat \Psi^*$ appears as a sum, the expression for $\Psi$ becomes a so-called inf-convolution where all reactions are strongly intertwined. For more details on these inf-convolutions and the factor $1/2$ in front of the free energy, we again refer to \cite[Sec.~3.4]{MielkePattersonPeletierRenger2017}.

\section{Leading example 2: fast-slow reaction fluxes}
\label{sec:reaction fluxes}

As mentioned in the introduction, the motivation behind the current paper is to search for thermodynamically consistent structures for systems that are not detailed balanced. The idea is that we increase the space by taking fluxes into account. However, in order to see the connection between structures in flux space and structures in state space, we dedicate this section to an example that is `almost' detailed balanced. More general systems will be dealt with when discussing the general theory in Section~\ref{sec:general theory}.

\subsection{Reacting particle system and macroscopic equation}

We consider a system of fast and slow chemical reactions, $\R=\R_\slow\cup\R_\fast$, where the slow reactions are assumed to be of mass-action kinetics and detailed balanced, as in the previous section. By contrast, we will not assume anything of the like for the fast reactions, neither shall we assume that each reaction consist of a forward and a backward reaction. However, in the microscopic model, we shall assume that the fast reactions happen on a faster time-scale, i.e.
\begin{equation}
  \mfrac{1}{V^2}\lambda\super{V}_r(\rho)\to \tilde k_r(\rho)\qquad \text{for } r\in\R_\fast,
\label{eq:RRE2 fast reaction scaling}
\end{equation}
but their effect on the concentrations is smaller, i.e. $\rho\super{V}(t)=\rho\super{V}(t^-) + \tfrac1{V^2}\gamma_r$ whenever a reaction $r\in\R_\fast$ occurs at time $t$. To notationally distinguish the two time scales we use a tilde ($\,\tilde{}\,$) to denote one-way fast quantities. Let us briefly mention that the additional fast reactions are not essential to see the relation between structures in flux and state space; they just lead to a richer example that is interesting in its own right.

The main idea is now to increase the state space by bookkeeping the events in the microscopic system, in this case, by counting the number of reactions that have occurred up to a given time. With the right scaling, this defines the following \emph{integrated reaction fluxes}\footnote{The term ``integrated'' signifies that these fluxes are cumulative over time. This simplifies the microscopic analysis since the corresponding process is Markovian; on a macroscopic scale only the time-derivatives will play a role. We also mention that we consider net rather than one-way slow fluxes, else the slow dynamics would not induce any force or gradient structure, see \cite[Sec.~4.6]{Renger2017}.}:
\begin{align*}
  &\bW\super{V}_{t,r} := \mfrac1V\#\big\{\text{forward reactions } r \text{ occurred in } (0,t)\big\} \\
  &\hspace{2cm} - \mfrac1V\#\big\{\text{backward reactions } r \text{ occurred in } (0,t)\big\},      &r\in\R_\slow,\\
  &\tilde W\super{V}_{t,r} := \mfrac1{V^2}\#\big\{\text{reactions } r \text{ occurred in } (0,t)\big\},      &r\in\R_\fast,
\end{align*}
and to shorten notation we sometimes write
\begin{equation}
  W\super{V}_t:=\big((\bW\super{V}_{t,r})_{r\in\R_\slow},(\tilde W\super{V}_{t,r})_{r\in\R_\fast}\big).
\label{eq:RRE2 flux vector notation}
\end{equation}
We shall always assume that the initial condition is known (deterministically) a priori, so that the concentrations can be retrieved from the integrated fluxes via the \emph{continuity equation}:
\begin{equation*}
  \rho\super{V}_t = \phi\super{V}\lbrack W\super{V}_t\rbrack := \rho\super{V}_0 + \Gamma W\super{V}_t.
\end{equation*}
In this sense the integrated fluxes encode more information than the concentrations.

The integrated fluxes are again a Markov process, with generator (cf.~\eqref{eq:RRE1 generator}):
\begin{align*}
  (\Q\super{V} f)(w) := &\sum_{r\in\R_\slow} 
     \lambda\super{V}_{r,\fw}\big(\phi\super{V}\lbrack w\rbrack\big)\big(  f(w+\tfrac1V\mathds1_r)- f(w) \big) \notag\\
   &\qquad+ \lambda\super{V}_{r,\bw}\big(\phi\super{V}\lbrack w\rbrack\big)\big(  f(w-\tfrac1V\mathds1_r)- f(w) \big) \notag\\
    &+ \sum_{r\in\R_\fast}
     \lambda\super{V}_r\big(\phi\super{V}\lbrack w\rbrack\big)\big(  f(w+\tfrac1{V^2}\mathds1_r)- f(w) \big),
\end{align*}

\subsection{Limit and large deviations}

We now mimic the arguments of Subsection~\ref{subsec:RRE1 dynamics limit ldp}, but now in the space of fluxes. Let $\rho\super{V}_0\to\rho_0$, and so $\phi\super{V}\lbrack w\rbrack\to\phi\lbrack w\rbrack:=\rho_0+\Gamma w$. Then by the same argument as in Subsection~\ref{subsec:RRE1 dynamics limit ldp}, using the scalings~\eqref{eq:RRE1 convergent reaction rates} and \eqref{eq:RRE2 fast reaction scaling}, one finds that as $V\to\infty$, the random process $W\super{V}_t$ converges (pathwise in probability) to the solution of the macroscopic equations:
\begin{align}
  \dot{\bar{w}}_{t,r} = \bk_r\big(\phi\lbrack w_t\rbrack\big), \quad r\in\R_\slow,
  &&\text{and}&&
  \dot{\tilde w}_{t,r} = \tilde k_r\big(\phi\lbrack w_t\rbrack\big), \quad r\in\R_\fast,
\label{eq:RRE2 limit eq}
\end{align}
again using the notation~\eqref{eq:RRE2 flux vector notation}. Indeed, combining these equations leads to the macroscopic equation~$\dot\rho_t=\Gamma k(\rho_t)$ for the concentrations.

Similarly, we study the fluctuations through the non-linear generator:
\begin{align}
  (\H\super{V} f)(w)
    &:=\tfrac1V \e^{-V f(w)}(\Q\super{V}\e^{V f})(w) \notag \\
    &=\sum_{r\in\R_\slow} \tfrac1V \lambda\super{V}_{r,\fw}\big(\phi\super{V}\lbrack w\rbrack\big)\big(\e^{V f(w+\tfrac1V\mathds1_r)-V f(w)}-1\big) \notag\\[-0.4cm]
    &\hspace{3cm} + \tfrac1V \lambda\super{V}_{r,\bw}\big(\phi\super{V}\lbrack w\rbrack\big)\big( \e^{V f(w-\tfrac1V\mathds1_r)-V f(w)}-1\big) \notag\\
    &\quad + \sum_{r\in\R_\fast} \tfrac1V \lambda\super{V}_r\big(\phi\super{V}\lbrack w\rbrack\big) \big( \e^{V f(w+\tfrac1{V^2}\mathds1_r)-V f(w)}-1\big) \notag\\
    &\hspace{-0.7cm}\xrightarrow{V\to\infty} \sum_{r\in\R_\slow} k_{r,\fw}\big(\phi\lbrack w\rbrack)\big) \big( \e^{\partial_{w_r} f(w)}-1\big) + k_{r,\bw}\big(\phi\lbrack w\rbrack\big) \big( \e^{-\partial_{w_r} f(w)}-1\big) \notag\\
    &\hspace{5cm}+ \sum_{r\in\R_\fast} \tilde k_r\big(\phi\lbrack w\rbrack\big) \partial_{w_r} f(w).
\label{eq:RRE2 nonlinear semigroup}
\end{align}
As in Section~\ref{subsec:RRE1 dynamics limit ldp}, this limit depends on the gradient $\grad f(w)$ only, which is consistent with the deterministic limit~\eqref{eq:RRE2 limit eq}. Again we set, by a slight abuse of notation,
\begin{align}
  \H(w,\zeta) &:= \sum_{r\in\R_\slow} k_{r,\fw}\big(\phi\lbrack w\rbrack\big) \big( \e^{\bar \zeta_r}-1\big) + k_{r,\bw}\big(\phi\lbrack w\rbrack\big) \big( \e^{-\bar\zeta_r}-1\big) \notag\\[-0.2cm]
      &\hspace{6cm} + \sum_{r\in\R_\fast} \tilde k_r\big(\phi\lbrack w\rbrack\big) \tilde \zeta_r,
\label{eq:RRE2 H}\\
  \L(w,j) &:= \sup_{\zeta\in\RR^\R} \zeta\cdot j - \H(w,\zeta) \notag\\
  &= \inf_{j_\fw-j_\bw=\bar\jmath} h\big(j_\fw | k_\fw(\phi\lbrack w\rbrack)\big) + h\big(j_\bw | k_\bw(\phi\lbrack w\rbrack)\big) + \chi_{\{\tilde\jmath = \tilde k(\phi\lbrack w\rbrack)\}},
\notag
\end{align}
using the notation $\chi_{\{\tilde\jmath = \tilde k(\phi\lbrack w\rbrack)\}}=0$ if $\tilde\jmath = \tilde k(\phi\lbrack w\rbrack)$ and $\infty$ otherwise.
Then, the dynamic large-deviation principle for the integrated fluxes $W\super{V}(t)$ state that (see \cite{PattersonRenger2018flux} for a rigorous proof):
\begin{equation}
  \Prob\!\big(W\super{V}_{(\cdot)}\approx w_{(\cdot)}\big) \stackrel{V\to\infty}{\sim}  \e^{-V\int_0^T\! \L(w_t,\dot w_t)\,dt}.
\label{eq:RRE2 dynamic ldp}
\end{equation}

Comparing this dynamic large-deviations principle with \eqref{eq:RRE1 dynamic ldp}, we see -- not coincidentally -- strong similarities. Let us assume that the limit fast fluxes do not influence the concentration, i.e.
\begin{equation}
  \Gamma\tilde k\big(\phi\lbrack w\rbrack\big)=0 \qquad\text{for all } w.
\label{eq:RRE2 decoupling}
\end{equation} 
Naturally, this macroscopic condition entails a sort of decoupling between the slow and fast dynamics, e.g. when the species involved in the slow dynamics act as a catalyst for the fast dynamics. In that case $\rho\super{V}_t=\phi\super{V}\lbrack W\super{V}_t\rbrack$ is exactly the process with generator~\eqref{eq:RRE1 generator}. The contraction principle of large deviations theory~\cite[Th.~4.2.1]{Dembo1998} then states that the two large-deviation costs are related via:
\begin{equation*}
  \int_0^T\!\hat\L(\rho_t,\dot\rho_t)\,dt = \inf_{w_{(\cdot)}:\rho_t=\phi\lbrack w_t\rbrack}\, \int_0^T\!\L(w_t,\dot w_t)\,dt.
\end{equation*}
In this setting, this infimum is only with respect to the second variable, because $\L(w,j)$ depends on $w$ through $\phi\lbrack w\rbrack$ only, which in turn arises naturally from the fact that the jump rates depend on the state and not on the integrated flux. Therefore the relation above simplifies further to:
\begin{equation}
  \hat\L\big(\phi\lbrack w\rbrack,s\big) = \inf_{j:s=\Gamma j} \L(w,j) \qquad\text{for all } w,s,
\label{eq:RRE2 hat L L}
\end{equation}
which is consistent with Remark~\ref{rem:RRE1 entropic form}. These relations will be the starting point of the general theory that we develop in Section~\ref{sec:general theory}. Let us only mention here that as a consequence we also have the relation $\hat\H\big(\phi\lbrack w\rbrack,\xi\big) = \H(w,\Gamma\tp\xi)$ for all $w,\xi$, cf.~\eqref{eq:RRE1 H} and \eqref{eq:RRE2 H}.

\subsection{Induced Generic structure}
\label{subsec:RRE2 induced ggen}

We now investigate whether the large deviations~\eqref{eq:RRE2 dynamic ldp} induces some structure in the space of fluxes. It turns out that this is indeed the case. As was found in \cite[Cor.~4.8]{Renger2017}, under the detailed balance assumption~\eqref{eq:RRE1 chem detailed balance} the slow fluxes induce the GGS $(\RR^{\R_\slow}_+,\Psi,\F)$, where
\begin{align}
  \F(w)&:=\hat\F\big(\phi\lbrack w\rbrack\big) = \tfrac12 h(\rho_0+\Gamma w|\rho^*), \label{eq:RRE2 F}\\
  \Psi^*(w,\bar \zeta)&:=\sum_{r\in\R_\slow}\sigma_r(w)\big(\cosh(\bar \zeta_r)-1\big)\qquad\text{and}\label{eq:RRE2 Psis}\\
  \Psi(w,\bar\jmath)&:=\sum_{r\in\R_\slow} \sigma_r(w) \Big( \cosh^*\!\big(\tfrac{\bar \jmath_r}{\sigma_r(w)}\big) +1\Big),
\label{eq:RRE2 Psi}
\end{align}
with $\sigma_r(w):=2\sqrt{k_{r,\fw}\big(\phi\lbrack w\rbrack\big)k_{r,\bw}\big(\phi\lbrack w\rbrack\big)}, r\in\R_\slow$. Extending these dissipation potentials to the full flux space by setting $\Psi^*(w,\zeta):=\Psi^*(w,\bar\zeta)$ and $\Psi(w,j):=\Psi(w,\bar j) + \chi_{\{\tilde j=0\}}$, we can decompose the large-deviation cost function as
\begin{equation*}
  \L(w,j) = \Psi\big(w,j-\tilde k(\phi\lbrack w\rbrack)\big) + \Psi^*\big(w,-\grad\F(w)\big) + \grad\F(w)\cdot\big(j-\tilde k(\phi\lbrack w\rbrack)\big),
\end{equation*}
and accordingly, the macroscopic evolution~\eqref{eq:RRE2 limit eq} as:
\begin{equation*}
  \dot w_t = \grad_\zeta\Psi^*\big(w_t,-\grad\F(w_t)\big) + \tilde k(\phi\lbrack w\rbrack).
\end{equation*}
Due to \eqref{eq:RRE2 F} and the decoupling condition~\eqref{eq:RRE2 decoupling}, the fast fluxes are orthogonal to the driving force, i.e.
\begin{equation}
  \grad_w\F(w)\cdot\tilde k(\phi\lbrack w\rbrack) = \grad_\rho\hat\F\big(\phi\lbrack w\rbrack\big)\cdot \Gamma\tilde k(\phi\lbrack w\rbrack)=0.
\label{eq:RRE2 pGGEN NIC}
\end{equation}
The quadruple $(\RR^\R_+,\Psi,\F,\tilde k\circ \phi)$ satisfying this condition falls within the class of what is recently coined \emph{pre-(Generalised) Generic (pGGEN)} \cite{KLMP2017phys,KLMP2018math}.

It was shown in those works that non-interaction condition~\eqref{eq:RRE2 pGGEN NIC} is a necessary and sufficient condition for the existence of an underlying \emph{Generalised Generic (GGEN)} structure $(\RR^\R_+,\Psi,\F,L,\E)$\footnote{One often needs to introduce an auxiliary energy (e.g. a heat bath) to force conservation of energy, which enlarges the degrees of freedom in the system. To keep notation accessible we ignore this issue.}.
This means that the fast flux term is Hamiltonian $\tilde k(\phi\lbrack w\rbrack)=L(w)\grad\E(w)$ for some Poisson structure $L$ satisfying the Jacobi identity (see Section~\ref{sec:general theory}), $\E$ is some Hamiltonian energy, and the following two non-interaction conditions are satisfied:
\begin{align}
  L(w)\grad\F(w) = 0 &&\text{and}&& \Psi^*\big(w,\zeta+ z\grad\E(w)\big) = \Psi^*(w,\zeta)
\label{eq:RRE2 NIC}
\end{align}
for all $w\in\RR^\R_+,\zeta\in\RR^\R$ and $z\in\RR$. These two conditions guarantee that along solutions the free energy is non-increasing and the Hamiltonian energy is conserved:
\begin{align*}
  \tfrac{\d}{\d t}\F(w_t) &= \underbrace{\grad\F(w)\cdot \grad_\zeta\Psi^*\big(w_t,-\grad\F(w)\big)}_{\leq 0 \text{ by convexity}} + \underbrace{\grad\F(w)\cdot L(w)}_{=0 \text{ by } \eqref{eq:RRE2 NIC}}\grad\E(w) \leq 0,\\
  \tfrac{\d}{\d t}\E(w_t) &= \underbrace{\grad\E(w)\cdot \grad_\zeta\Psi^*\big(w_t,-\grad\F(w)\big)}_{=0 \text{ by } \eqref{eq:RRE2 NIC}} + \underbrace{\grad\E(w)\cdot L(w)\grad\E(w)}_{=0 \text{ by skewsymmetry}} =0.
\end{align*}

One main message of this section is that the flux large-deviation cost $\L$ induces a unique pGGEN system~$(\RR^\R_+,\Psi,\F,\tilde k\circ \phi)$. Although this implies the existence of a GGEN system induced by $\L$, one can not uniquely decide on the basis of $\L$ what the `correct' Hamiltonian structure $(\RR^\R_+,L,\E)$ should be. Additional physical or mathematical arguments needed to uniquely fix the Hamiltonian structure are beyond the scope of this paper; for possible constructions of Poisson structures $L$ and energies $\E$, we refer the reader to \cite[Sec.~4]{KLMP2018math}).

It should be noted that the pGGEN system~$(\RR^\R_+,\Psi,\F,\tilde k\circ \phi)$ is rather special in that \eqref{eq:RRE2 pGGEN NIC} holds for any $\hat\F$, since the drift $\tilde k\circ \phi$ is `divergence-free'. We will see that such systems play a special role in connecting systems in flux and state space.  

Another main message is that the GGS part~\eqref{eq:RRE2 F} on flux space is strongly related to the GGS~\eqref{eq:RRE1 GGS} on state space, just like both cost functions are related by \eqref{eq:RRE2 hat L L}. These observations will be the basis of the general theory of the next section.

\section{General theory}
\label{sec:general theory}

Following the examples of Sections~\ref{sec:chemical reactions} and \ref{sec:reaction fluxes}, we now develop a more abstract framework to study to the relation between energy-driven structures in flux and in state space.

\subsection{Geometry and notation}

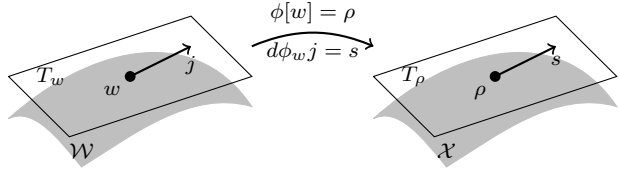
\begin{figure}[h!]
\centering
\begin{tikzpicture}[scale=0.8,font=\footnotesize]
\filldraw[color=lightgray] (0,-0.7)..controls (1.2,0.4) and (3,1)..
                           (4,-0.5)..controls (3,-0.2) and (2,-1) ..
                           (1.2,-1.5) node[anchor=south]{\color{black}$\W$}.. controls (0.8,-1) and (0.4,-0.5) .. (0,-0.7);
\draw(0,0) node[right=7]{$T_w$}--(3,1)--(4,0)--(1,-1)--(0,0);
\filldraw(2,0) node[anchor=north east]{$w$} circle(0.08);
\draw[thick,->](2,0)--(3,0.5) node[anchor=north]{$j$};

\filldraw[color=lightgray] (6,-0.7)..controls (7.2,0.4) and (9,1)..
                           (10,-0.5)..controls (9,-0.2) and (8,-1) ..
                           (7.2,-1.5) node[anchor=south]{\color{black}$\X$}.. controls (6.8,-1) and (6.4,-0.5) .. (6,-0.7);
\draw(6,0) node[right=7]{$T_\rho$}--(9,1)--(10,0)--(7,-1)--(6,0);
\filldraw(8,0) node[anchor=north east]{$\rho$} circle(0.08);
\draw[thick,->](8,0)--(9,0.5) node[anchor=north]{$s$};

\draw[thick,->](4,0.5)..controls (4.7,0.8) and (5.3,0.8) .. node[midway,anchor=south]{$\phi\lbrack w\rbrack=\rho$} node[midway,anchor=north]{$d\phi_w j=s$}(6,0.5);
\end{tikzpicture}
\caption{The `continuity' mapping between the flux and state manifolds.}
\label{fig:manifolds}
\end{figure}

Throughout this section we assume to be given:

\begin{itemize}
\item A differentiable manifold $\W$ (``flux space''), where tangents are denoted by $(w,j)\in T\W$ and cotangents by $(w,\zeta)\in T^*W$. Note that we distinguish between tangents and cotangents; we write the dual pairing between them as ${}_{T_w^*}\!\langle\zeta,j\rangle_{T_w}$ or simply $\langle\zeta,j\rangle$; 
\item A differentiable manifold $\X$ (``state space''), where tangents are denoted by $(\rho,s)\in T\X$ and cotangents by $(\rho,\xi)\in T^*\X$;
\item A surjective differentiable operator $\phi:\W\to\X$ with bounded linear differential $d\phi_w :T_w\to T_{\Gamma(w)}$ and adjoint differential $d\phi_w \tp: T_{\phi\lbrack w\rbrack}^*\to T_w^*$. This defines an abstract continuity equation $\phi\lbrack w\rbrack =\rho$, or in differentiated form $d\phi_w j = s$, see Figure~\ref{fig:manifolds}. Contrary to Sections~\ref{sec:chemical reactions} and \ref{sec:reaction fluxes}, the second continuity equation may now also depend on $w$. In practice, the continuity mapping $\phi\lbrack w\rbrack$ depends on the initial state $\rho_0$, which we assume to be fixed.
\item An L-function $\L:T\W\to\RR_+$ on flux space (see below for the definition of L-functions). This function could be a dynamic large-deviation cost function corresponding to random fluxes in some microscopic system, but throughout this section it could also be a more general expression.
\item An L-function $\hat\L:T\X\to\RR_+$ on state space, related to the flux space L-function via
\begin{equation}
  \hat\L(\rho,s):=\inf_{\substack{(w,j)\in T\W:\\\phi\lbrack w\rbrack = \rho, \,d\phi\lbrack w\rbrack j=s}} \L(w,j).
\label{eq:gent hatL=inf L}
\end{equation}
This relation is again inspired by large-deviation theory, where such relation holds due to the contraction principle~\cite[Th.~4.2.1]{Dembo1998}.
\item Corresponding to the L-functions are their convex duals with respect to their second variable, i.e. $\H:T^*\W\to\RR$ and $\hat\H:T^*\X\to\RR$ with
\begin{align*}
  \H(w,\zeta):=\sup_{j\in T_w} \langle \zeta,j\rangle - \L(w,j)
  &&\text{and}&&
  \hat\H(\rho,\xi):=\sup_{s\in T_\rho} \langle \xi,s\rangle - \L(\rho,\xi).
\end{align*}
We express assumptions in terms of these duals, since in practice they are often more explicitly given than their corresponding L-functions.
\end{itemize}

Recall that in the previous section we saw that $\L(w,j)$ depends on $w$ through $\rho=\phi\lbrack w\rbrack$ only. This condition becomes slightly more complicated in non-flat spaces, for a number of reasons. Firstly, the flux $j$ can not be kept fixed while changing $w$, unless one has a path-independent notion of parallel transport, i.e. the space is flat. Secondly, even if $\L(w,j)$ would not depend on $w$, this dependence could re-enter through the continuity equation in the infimum $\inf_{j:d\phi_w j=s}\L(w,j)$. Therefore, the condition that we need is that
$  \hat\L(\phi\lbrack w\rbrack,s)=\inf_{\substack{j\in T_w:\\d\phi_w j=s}} \L(w,j)$
for all $w\in\W, s\in T_{\phi\lbrack w\rbrack}$. It is easily seen that this condition is equivalent to the following flux invariance:
\begin{quoting}[rightmargin=\parindent]
for fixed $(\rho,\xi)\in T^*\X$, the function $\phi^{-1}\lbrack\rho\rbrack\ni w\mapsto\H(w,d\phi_w\tp\xi)$ does not depend on $w$.
\eqnum
\label{eq:gent invariance}
\end{quoting}

All manifolds and functionals are assumed to be sufficiently differentiable wherever needed. For a (differentiable) functional $F:\W\to\RR$ (and similarly on flux space) we write $dF(w)\in T_w^*$ for the derivative, in the sense that on a curve $\tfrac{d}{dt} F(w_t)=\langle dF(w_t),\dot w_t\rangle$.

\subsection{Definitions}

We now define the notions of L-functions, dissipation potentials, GGS, pGGEN, Poisson operators and GGEN on flux space; the same concepts on state space are defined analogously. Naturally all notions are compatible with the exposition from Section~\ref{sec:chemical reactions} and \ref{sec:reaction fluxes}.

\begin{definition} We say $\L:T\W\to\RR_+$ is an L-function whenever for all $w\in\W$:
\begin{enumerate}[(i)]
  \item $j\mapsto\L(w,j)$ is convex;
  \item $\L(w,j)=0 \iff j=\A_\L(w)$ for some unique vector field $\A_\L$.
\end{enumerate}
\end{definition}
Due to the convexity, $\L$ is also the convex dual of $\H$. Central to GGS, pGGEN and GGEN is the notion of dissipation potentials:

\begin{definition} A function $\Psi:T\W\to\RR_+$ is called a \emph{dissipation potential} whenever for all $w\in\W$:
\begin{enumerate}[(i)]
\item $j\mapsto\Psi(w,j)$ is convex;
\item $\Psi(w,0)=0$;
\end{enumerate}
If these conditions hold, then the same conditions hold for the (pre-)dual dissipation potential
\begin{equation}
  \Psi^*(w,\zeta):=\sup_{j\in T_w}\langle\zeta,j\rangle - \Psi(w,j).
\label{eq:gent Psis from Psi}
\end{equation}
We also say that $(\Psi,\Psi^*)$ is a \emph{dissipation potential} pair whenever $\Psi$ is a dissipation potential.
\end{definition}

\begin{definition} A \emph{generalised gradient system} (GGS) is a triple $(\W,\Psi,\F)$, where $\W$ is a differentiable manifold, $\F:\W\to\RR$ and $\Psi:T\W\to\RR_+$ is a dissipation potential.
We say that an L-function $\L$ induces a GGS $(\W,\Psi,\F)$ whenever for all $(w,j)\in T\W$:
\begin{equation}
  \L(w,j)=\Psi(w,j)+\Psi^*\big(w,-d\F(w)\big) + \langle d\F(w),j\rangle.
\label{eq:gent GGS L=PsiPsi*}
\end{equation}
\end{definition}

As explained in Subsection~\ref{subsec:RRE2 induced ggen}, by extending a GGS with an orthogonal drift one arrives at
\begin{definition}[\cite{KLMP2018math}] A \emph{Generalised pre-Generic system} (pGGEN) is a quadruple $(\W,\Psi,\F,b)$, where $\W$ is a differentiable manifold, $\Psi$ is a dissipation potential, $\F:\W\to\RR$, and $b(w)\in T_w$ is a vector field such that:
\begin{equation*}
  \langle d\F(w), b(w)\rangle =0 \qquad \text{for all } w\in\W.
\end{equation*}
We say that an L-function induces a pGGEN $(\W,\Psi,\F,b)$ whenever for all $(w,j)\in T\W$:
\begin{equation*}
  \L(w,j)=\Psi\big(w,j-b(w)\big)+\Psi^*\big(w,-d\F(w)\big) + \langle d\F(w),j\rangle.
\end{equation*}
\end{definition}

Finally, if the drift has the form of an Hamiltonian system that behaves more or less independently of the GGS part we arrive at a Generalised Generic system. In order to define this we first define:
\begin{definition} A linear operator $L:T^*\W\to T\W$ is called a Poisson structure if it satisfies the Jacobi identity
  \begin{equation*}
    \big\{\{\F_1,\F_2\}_L,\F_3\big\}_L+\big\{\{\F_2,\F_3\}_L,\F_1\big\}_L+\big\{\{\F_3,\F_1\}_L,\F_2\big\}_L=0
  \end{equation*}
for all $\F_{1,2,3}:\W\to\RR$, where $\{\F_1,\F_2\}_L(w):=\langle d\F_1(w),L(w) d\F_2(w)\rangle$; 
\end{definition}
Jacobi's identity implies skew symmetry, i.e. $\langle\zeta_1,L(w)\zeta_2\rangle = - \langle\zeta_2,L(w)\zeta_1\rangle$ for all $w\in\W$, $\zeta_1,\zeta_2\in T_w$; in particular one has $\langle\zeta,L(w)\zeta\rangle=0$. Finally,
\begin{definition}[{\cite[Sect.~2.5]{Mielke2011GGEN}}] A \emph{generalised Generic system} (GGEN) is a quintuple $(\W,\Psi,\F,L,\E)$, where $\W$ is a differentiable manifold, $\Psi$ is a dissipation potential, $\E,\F:\W\to\RR$, $L:T^*\W\to T\W$ is a Poisson structure, and the two non-interaction conditions are satisfied:
\begin{align}
  &L(w)d\F(w)=0 \qquad\text{for all }w\in\W, \text{ and}
    \label{eq:gent NIC Psi*DE}\\
  &\Psi^*\big(w,\zeta+ \lambda d\E(w)\big) = \Psi^*(w,\zeta) \qquad\text{for all } (w,\zeta)\in T^*\W \text{ and } \lambda\in\RR.
    \label{eq:gent NIC LDF}
\end{align}
We say that an L-function induces a GGEN $(\W,\Psi,\F,L,\E)$ whenever for all $(w,j)\in T\W$:
\begin{equation}
  \L(w,j)=\Psi\big(w,j-L(w)d\E(w)\big)+\Psi^*\big(w,-d\F(w)\big) + \langle d\F(w),j\rangle.
\label{eq:GGEN L=PsiPsi*}
\end{equation}
\end{definition}

\begin{remark} If a GGS, pGGEN or GGEN is given on a manifold $\W$, and the dissipation potentials are quadratic (as explained in Section~\ref{subsec:RRE1 GGS}), then one can use the positive definite operator $K(w)\xi:=d_\zeta\Psi^*(w,\zeta)$ to define a new manifold, and redefine everything on this manifold. This allows to study the structures from a more geometric point of view.
\label{rem:new manifold}
\end{remark}

\subsection{From L-functions to GGS, pGGEN and GGEN}

We now recall some of the main results from \cite{MielkePeletierRenger2014} and \cite{KLMP2018math}, that give necessary and sufficient conditions for an L-function to induce a GGS or pGGEN. A similar result for GGEN does not exist, since an L-function does not uniquely determine a Poisson operator and Hamiltonian energy. However, from a pGGEN one can always construct a GGEN (in a non-unique way); for that result we refer the reader to \cite[Sec.~4]{KLMP2018math}.

Again, the following results are described, but not restricted to flux space.

\begin{theorem}[{\cite[Prop.~2.1 \& Th.~2.1]{MielkePeletierRenger2014}}] Let $\L:T\W\to\RR_+$ be an L-function with convex dual $\H$, and let $\F:\W\to\RR$ be given. Then the following statements are equivalent:
\begin{enumerate}[(i)]
\item $\L$ induces a GGS $(\W,\Psi,\F)$ for some dissipation potential $\Psi$,
\item $\H(w,\zeta) = \Psi^*\big(w,\zeta-d_w\F(w)\big) - \Psi^*\big(w,-d\F(w)\big)$ for some diss. pot. $\Psi^*$,
\item $d_\zeta\H\big(w,d_w\F(w)\big)=0$.
  \eqnum\label{eq:gent MPR condition on H}
\item $d_j\L(w,0)=d_w\F(w)$,
  \eqnum\label{eq:gent MPR condition on L}
\end{enumerate}
In that case $\Psi^*$ (and indirectly $\Psi$) is uniquely determined by
\begin{equation}
  \Psi^*(w,\zeta)=\H\big(w,\zeta+d_w\F(w)\big) - \H\big(w,d_w\F(w)\big).
\label{eq:MPR Psis from H}
\end{equation}
\label{th:gent MPR}
\end{theorem}
From condition~\eqref{eq:gent MPR condition on L} we see that $\F$, if it exists, is uniquely given up to constants. Note in particular that conditions~\eqref{eq:gent MPR condition on H} and \eqref{eq:gent MPR condition on L} do not involve $\Psi$.

\begin{theorem}[{\cite[Th.~3.6]{KLMP2018math}}]

Let $\L:T\W\to\RR_+$ be an L-function with convex dual $\H$, and let $\F:\W\to\RR$ and a vector field $b(w)\in T_w$ be given for which $\langle d_w\F(w),b(w)\rangle=0$. Then the following statements are equivalent:
\begin{enumerate}[(i)]
\item $\L$ induces a pGGEN $(\W,\Psi,\F,b)$ for some dissipation potential $\Psi$,
\item $\H(w,\zeta) = \Psi^*\big(w,\zeta-d_w\F(w)\big) - \Psi^*\big(w,-d_w\F(w)\big)+\langle \zeta,b(w)\rangle$ for some diss. pot. $\Psi^*$,
  \eqnum\label{eq:gent KLMP H induce pGGEN}
\item $d_\zeta\H\big(w,d_w\F(w)\big)=b(w)$,
\item $d_j\L\big(w,b(w)\big)=d_w\F(w)$.
  \eqnum\label{eq:gent KLMP condition on L}
\end{enumerate}
In that case $\Psi^*$ is uniquely determined by
\begin{equation}
  \Psi^*(w,\zeta)=\H\big(w,\zeta+d_w\F(w)\big) - \H\big(w,d_w\F(w)\big) - \langle \zeta,b(w)\rangle.
\label{eq:gent KLMP Psis from H}
\end{equation}
\label{th:gent KLMP}
\end{theorem}

From condition~\eqref{eq:gent KLMP H induce pGGEN} we see that $\H$ must consist of a convex part and a linear part $\langle \zeta,b(w)\rangle$, so that the drift $b$ is a priori and uniquely fixed by $\H$. Therefore $\F$ is again uniquely fixed by condition~\eqref{eq:gent KLMP condition on L}. This is different from the Generic setting; $Ld\E$ uniquely defines $d\F$ and vice versa, but the whole quintuple may not be unique. However, one can still state a Generic analogue of Theorems~\ref{th:gent MPR} and \ref{th:gent KLMP} as follows:

\begin{proposition} Let $\L:T\W\to\RR_+$ be an L-function with convex dual $\H$, and let a Poisson structure $L:T^*\W\to T\W$ and energies $\E,\F:\W\to\RR$ be given such that the non-interaction condition~$Ld\F=0$ holds. Then the following statements are equivalent:
\begin{enumerate}[(i)]
\item $\L$ induces a GGEN $(\W,\Psi,\F,L,\E)$ for some dissipation potential $\Psi$,
\item $d_s\L\big(w,L(w) d_w\E(w)\big)=d_w\F(w)$ and
  \begin{equation}
    \L(w,j)=\infty \qquad \text{ for all } j\in T_w \text{ for which } \langle d\E,j\rangle\neq0,
    \label{eq:gent NIC condition on L}    
  \end{equation}
\item $d_\zeta\H\big(w,d_w\F(w)\big)=L(w)d_w\E(w)$ and
\eqnum\label{eq:gent GGEN condition on H}
  \begin{equation}
    \H\big(w,\zeta+\lambda d\E(w)\big) = \H(w,\zeta) \qquad\text{for all } (w,\zeta)\in T^*\W \text{ and } \lambda \in\RR.
    \label{eq:gent NIC condition on H}
  \end{equation}
\end{enumerate}
In that case $\Psi^*$ is uniquely determined by
\begin{equation}
  \Psi^*(w,\zeta)=\H\big(w,\zeta+d_w\F(w)\big) - \H\big(w,d_w\F(w)\big).
\label{eq:gent MPR GGEN Psis from H}
\end{equation}
\end{proposition}

\begin{proof} By Lemma~\ref{lem:gent shift GGS into GGEN} (see below), we can apply Theorem~\ref{th:gent MPR} to the shifted L-function $\tilde\L(w,j):=\L\big(w,j+L(w)d_w\E(w)\big)$ and back. This yields the three equivalences, apart from the other non-interaction condition~\eqref{eq:gent NIC Psi*DE}. From the explicit formula~\eqref{eq:gent MPR GGEN Psis from H} one finds that the missing non-interaction condition is equivalent to~\eqref{eq:gent NIC condition on L} and to~\eqref{eq:gent NIC condition on H}.
\end{proof}

The previous proof made use of the following lemma:

\begin{lemma} Let $L:T^*\W\to T\W$ be a Poisson operator, $\E,\F:\W\to\RR$ be energies and $\Psi$ be a dissipation potential such that the non-interaction conditions~\eqref{eq:gent NIC Psi*DE} and \eqref{eq:gent NIC LDF} hold. An L-function $\L$ induces an GGEN $(\W,\Psi,\F,L,\E)$ if and only if the shifted L-function $\tilde\L(w,j):=\L\big(w,j+L(w)d_w\E(w)\big)$ induces the GGS $(\W,\Psi,\F)$.
\label{lem:gent shift GGS into GGEN}
\end{lemma}

\begin{proof} Because of the non-interaction condition~\eqref{eq:gent NIC LDF}, the shift transforms relation~\eqref{eq:GGEN L=PsiPsi*} into \eqref{eq:gent GGS L=PsiPsi*}, and analogously for the other direction.
\end{proof}

\subsection{Relation between structures in flux and state space}

We now consider L-functions $\L$ and $\hat\L$ on flux and state space, and study how their induced structures are related.

\begin{proposition} Assume that an L-function $\L:T\W\to\RR_+$ induces a pGGEN $(\W,\Psi,\F,b)$ where $d\phi_w b(w) =0$ and
\begin{equation*}
 \hat\F\big(\phi\lbrack w\rbrack\big)=\F(w), \qquad\text{(up to constants)}
\end{equation*}
for some $\hat\F:\X\to\RR$. Then the L-function $\hat\L$ given by \eqref{eq:gent hatL=inf L} induces a GGS $(\X,\hat\Psi,\hat\F)$ for some dissipation potential $\hat\Psi$.
\label{prop:gent flux pGGEN then state GGS}
\end{proposition}
\begin{proof}
Since $d_w\hat\F(\phi\lbrack w\rbrack)=d\phi_w\tp d_\rho\hat\F(\phi\lbrack w\rbrack)$ and $d\phi_w b(w) =0$, we can rewrite
\begin{equation*}
  \hat\L(\rho,s) = \inf_{w\in \W:\phi\lbrack w\rbrack=\rho} \Big( \inf_{j\in T_w: d\phi_w j=s} \Psi(w,j)\Big) + \Psi^*\big(w,-d\phi_w\tp d_\rho\hat\F(\rho)\big) + \langle d_\rho\hat\F(\rho),s \rangle,
\end{equation*}
and because $\Psi,\Psi^*$ is a dissipation potential pair, clearly
\begin{equation*}
  \hat\L(\rho,s) - \langle d_\rho\hat\F(\rho),s \rangle\geq \inf_{w\in \W:\phi\lbrack w\rbrack=\rho} \Psi^*\big(w,-d\phi_w\tp d_\rho\hat\F(\rho)\big)         = \hat \L(\rho,0).
\end{equation*}
This is equivalent to $d_\rho\hat\F(\rho)\in\partial_s\hat\L(\rho,0) = \{d_s\hat\L(\rho,0)\}$, which by Theorem~\ref{th:gent MPR} implies that $\hat\L$ induces a GGS $(\X,\hat\Psi,\hat\F)$ for some $\hat\Psi$. 
\end{proof} In the above proposition, $\F$ also on $w$ through $\phi\lbrack w\rbrack$ only, which is a very physical assumption. It does imply however, that the equilibria of the flux gradient system can only be 
unique up to the kernel of $\phi$; this kernel can be interpreted as a generalisation of divergence-free vector fields.

A natural question is now whether we can turn the statement of Proposition~\ref{prop:gent flux pGGEN then state GGS} around.  Indeed, if the invariance condition~\eqref{eq:gent invariance} holds and we restrict to pGGEN with 'divergence-free drifts', then the statement becomes an equivalence, and we have an explicit relation between the flux and state dissipation potentials. This is a stronger version of the statement in~\cite[Prop.~4.7]{Renger2017}, where we related GGS to so-called `force structures'.

\begin{theorem} Assume that an L-function $\L:T\W\to\RR_+$ with corresponding dual $\H$ satisfies the invariance condition~\eqref{eq:gent invariance}, and let the L-function $\hat\L:\X\to\RR_+$ be given by~\eqref{eq:gent hatL=inf L}. Then $\L$ induces a pGGEN $(\W,\Psi,\F,b)$ with $d\phi_w b(w)=0$ and $\hat\F\circ\phi$ for some $\hat\F:\X\to\RR$ if and only if $\hat\L$ induces a GGS $(\X,\hat\Psi^*,\hat\F)$. In that case the dissipation potentials $\hat\Psi$ and $\hat\Psi^*$ are related to $\Psi$ and $\Psi^*$ through
\begin{subequations}
\begin{align}
  &\hat\Psi\big(\phi\lbrack w\rbrack,s\big)=\inf_{\substack{j\in T_w:\\d\phi_w j = s}} \Psi(w,j)
    \label{eq:gent hat Psi from Psi} \\
  &\hat\Psi^*\big(\phi\lbrack w\rbrack,\xi\big)=\Psi^*(w,d\phi_w\tp\xi).
    \label{eq:gent hat Psis from Psis}
\end{align}
\label{eq:gent hat diss from diss}
\end{subequations}
\label{th:gent flux pGGEN equiv state GGS}
\end{theorem}

\begin{proof} Assume that $\L$ induces a pGGEN $(\W,\Psi,\F,b)$ with $d\phi_w b(w)=0$ and $\hat\F\circ\phi$. Since by assumption $\H(w,d\phi_w\tp\xi)$ does not depend on $w\in\phi^{-1}\lbrack\rho\rbrack$, by~\eqref{eq:MPR Psis from H} the expression $\Psi^*(w,d\phi_w\tp\xi)$ is also invariant under this choice. Therefore we can define $\hat\Psi^*(\rho,\xi)$ by \eqref{eq:gent hat Psis from Psis}; it is easily checked that it its convex dual is given by \eqref{eq:gent hat Psi from Psi}. We can write:
\begin{align*}
  \hat\L(\rho,s) &= \inf_{d\phi_w j =s} \Big\{ \Psi\big(w,j-b(w)\big) + \Psi^*\big(w,-d_w\F(w)\big) + \langle d_w\F(w),j\rangle \Big\} \\
    &=\inf_{d\phi_w j =s} \Psi(w,j) + \Psi^*\big(w,-d\phi_w\tp d_\rho\hat\F(\rho)\big) + \langle d_\rho\hat\F(\rho),s\rangle \\
    &=\hat\Psi(\rho,s) + \hat\Psi^*\big(\rho,-d_\rho\hat\F(\rho)\big) + \langle d_\rho\hat\F(\rho),s \rangle,
\end{align*}
and hence $\hat\L$ induces the GGS $(\X,\hat\Psi,\hat\F)$, which is unique by Theorem~\eqref{th:gent MPR}.

For the other direction, assume that $\hat\L$ induces a GGS $(\X,\hat\Psi^*,\hat\F)$. Define $b(w):=d_\zeta\H\big(w,d_w\F(w)\big)$. Then by the invariance condition $\H(w,d\phi_w\tp\xi)  = \hat \H(\rho,\xi)$ and by \eqref{eq:gent MPR condition on H}:
\begin{multline*}
  d\phi_w b(w) = d\phi_w d_\zeta\H\big(w,d_w\F(w)\big) = d\phi_w d_\zeta\H\big(w,d\phi_w d_\rho\hat\F(\rho)\big) \\ = d_\xi\H\big(\rho,d_\rho\hat\F(\rho)\big)=0.
\end{multline*}
Now define $\Psi^*$ by \eqref{eq:gent KLMP Psis from H}. In particular $\Psi^*\big(w,-d\F(w)\big) = - \H\big(w,d_w\F(w)\big)$ since $\langle d_w\F(w),b(w)\rangle =0$ and, by the definition of L-functions, $\H(w,0)=0$. Then~\eqref{eq:gent KLMP H induce pGGEN} holds and hence by Theorem~\ref{th:gent KLMP} the flux L-function $\L$ induces the pGGEN $(\W,\Psi,\F,b)$.

\end{proof}

Due to the non-uniqueness of induced GGEN systems, there is no similar `if and only if' statement for the Generic setting. Nevertheless, in one direction, the GGEN analogue of Theorem~\ref{th:gent flux pGGEN equiv state GGS} is:

\begin{proposition} Assume that an L-function $\L:T\W\to\RR$ induces a GGEN $(\W,\Psi,\F,L,\E)$ where
\begin{align}
  \hat\F\big(\phi\lbrack w\rbrack\big)=\F(w) &&\text{and}&& 
  \hat\E\big(\phi\lbrack w\rbrack\big)=\E(w),
    \qquad\text{(up to constants)}
    \label{eq:hat F hat E from F E}
\end{align}
for some $\hat\F,\hat\E:\X\to\RR$, and that
\begin{equation}
  d\phi_w L(w)d\phi_w\tp=:\hat L(\rho) \text{ depends on $w$ through $\rho=\phi\lbrack w\rbrack$ only}.
\label{eq:hat Poisson}
\end{equation}
Then the L-function $\hat\L(\rho,s)$ given by \eqref{eq:gent hatL=inf L} induces a GGEN $(\X,\hat\Psi,\hat\F,\hat L,\hat\E)$ for some dissipation potential $\hat\Psi$.

If in addition, $\H$ satisfies the invariance principle~\eqref{eq:gent invariance}, then $\hat\Psi$ and $\hat\Psi^*$ are related to $\Psi$ and $\Psi^*$ through \eqref{eq:gent hat Psi from Psi} and \eqref{eq:gent hat Psis from Psis}.
\end{proposition}

\begin{proof} We again apply Lemma~\ref{lem:gent shift GGS into GGEN} to transform the problem into a problem of GGSs. Indeed, the L-function $\tilde\L(w,j):=\L\big(w,j+L(w)d_w\E(w)\big)$ induces the GGS $(\W,\Psi,\F)$. Hence by Proposition~\ref{prop:gent flux pGGEN then state GGS}, a GGS $(\X,\hat\Psi,\hat\F)$ (for some $\hat\Psi$) is induced by the L-function
\begin{align*}
  \hat{\tilde \L}(\rho,s) &:= \inf_{\phi\lbrack w\rbrack=\rho}\,\inf_{d\phi_w j=s} \tilde\L(w,j) \\
    &=\inf_{\phi\lbrack w\rbrack=\rho}\,\inf_{d\phi_w(j-L(w)d_w\E(w))=s} \L\big(w,j\big) \\
    &=\hat\L\big(\rho,s+\hat L(\rho)d_\rho\E(\rho)\big).
\end{align*}
If we can now validate that $\hat L$ is a Poisson structure, and that the non-interaction conditions are satisfied for $(\X,\hat\Psi,\hat\F,\hat L,\hat\E)$, then Lemma~\ref{lem:gent shift GGS into GGEN} concludes the proof.

For the Poisson structure, note that, for any smooth $\hat\F_1,\hat\F_2:\X\to\RR$, the Lie bracket remains unaltered:
\begin{align*}
  \{\hat\F_1,\hat\F_2\}_{\hat L}\big(\phi\lbrack w\rbrack\big)
    &=\{\F_1\circ\phi,\F_2\circ\phi\}_L(w).
\end{align*}
The non-interaction condition~\eqref{eq:gent NIC LDF} is clearly satisfied as for any $w\in\W$ we have
\begin{equation*}
  \hat L\big(\phi\lbrack w\rbrack\big)d_\rho\hat\F\big(\phi\lbrack w\rbrack\big) = d\phi_w L(w) d\phi_w\tp d_\rho\hat\F\big(\phi\lbrack w\rbrack\big) = d\phi_w L(w) d_w\F(w) = 0.
\end{equation*}
To check the other non-interaction condition~\eqref{eq:gent NIC Psi*DE} we use the equivalent formulation~\eqref{eq:gent NIC condition on H}. Indeed, for any $(\rho,\xi)\in T^*\X$ and $\lambda\in\RR$,
\begin{multline*}
  \hat\H\big(\rho,\xi + \lambda d_\rho d\hat\E(\rho)\big) = \sup_{\phi\lbrack w\rbrack=\rho} \H\Big(w,d\phi_w\tp(\xi + \lambda d_\rho\hat\E(\phi\lbrack w\rbrack)\big)\Big)\\
    = \sup_{\phi\lbrack w\rbrack=\rho} \H\big(w,d\phi_w\tp\xi + \lambda d_w\E(w)\big)
    = \sup_{\phi\lbrack w\rbrack=\rho} \H(w,d\phi_w\tp\xi) = \hat\H\big(\rho,\xi).
\end{multline*}

Finally, if $\H$ satisfies the invariance property~\eqref{eq:gent invariance}, then Proposition~\ref{prop:gent flux pGGEN then state GGS} yields relation~\eqref{eq:gent hat Psi from Psi} and \eqref{eq:gent hat Psis from Psis}.
\end{proof}

Condition~\eqref{eq:hat Poisson} is in a sense a natural one as the following result shows:
\begin{proposition} Assume that L-functions $\L$ and $\hat\L$ induce two GGENs $(\W,\Psi,\F,L,\E)$ and $(\X,\hat\Psi,\hat\F,\hat L,\hat\E)$, where $\F,\E$ are related to $\hat\F,\hat\E$ by \eqref{eq:hat F hat E from F E}, and $\H$ satisfies the invariance property~\eqref{eq:gent invariance}. Then
\begin{equation*}
  \hat L\big(\phi\lbrack w\rbrack\big)d_\rho\hat\E\big(\phi\lbrack w\rbrack\big) = d\phi_w L(w)d\phi_w\tp d_\rho\hat\E\big(\phi\lbrack w\rbrack\big).
\end{equation*}
\end{proposition}

\begin{proof} For any $w\in\W$ and $\rho=\phi\lbrack w\rbrack$, we may write $\hat\H(\rho,\xi) = \H(w,d\phi_w\tp\xi)$, and so by \eqref{eq:gent GGEN condition on H},
\begin{multline*}
  \hat L(\rho)d_\rho\hat\E(\rho)
    = d_\xi\hat\H\big(\rho,-d_\rho\hat\F(\rho)\big)
    = d\phi_w d_\zeta\H\big(w,-d\phi_w d_\rho\hat\F(\phi\lbrack w\rbrack)\big) \\
    = d\phi_w L(w) d\phi_w d_\rho\hat\E(\rho).
\end{multline*}

\end{proof}


\section{Diffusion}
\label{sec:diffusion}

In this section we apply the ideas of the previous section to a model for diffusion. The flux structure related to diffusion is interesting in its own right, and as far as the author is aware, previously unknown. In the next section we show how this model can be coupled with the results of Section~\ref{sec:reaction fluxes} to obtain flux and state GGSs/pGGENs for reaction-diffusion systems.

Typical microscopic models of diffusion consist of Brownian particles, or discretised versions thereof, like random walkers or an exclusion process. Since empirical fluxes are a bit easier to define on a lattice, we focus on independent random walkers\footnote{With the scaling that we use, the system of independent random walkers is `exponentially equivalent' to a system of Brownian motions, meaning they share the same hydrodynamic limit and large deviations.}.

\subsection{Diffusing particle system}

This microscopic actually has two scaling parameters: the number of particles, which we denote by $V$ for consistency with the rest of this paper, and the lattice spacing $\epsilon_V$. The speed with which $\epsilon_V\to0$ as $V\to\infty$ is irrelevant. For fixed $V$, let $\big(X_{t,i}\big)_{i=1}^{V}$ be independent random walkers on the lattice $(\epsilon_V\Z)^d$ with jump rate $\epsilon^{-2}$. Define the random concentration and (integrated, net) flux by:
\begin{align*}
  \rho\super{V}_{t}(dx)&:=\mfrac1V\#\{i=1,\hdots,V: X_{t,i}\in dx\}, \qquad \text{and}\\
  \bW\super{V}_{t,l}(dx)&:= \mfrac{\epsilon_V}{V}\#\big\{\text{jumps } \tilde x \text{ to } \tilde x + \epsilon_V \mathds1_l \text{ occurred in } (0,t): \tilde x + \tfrac{1}{2}\epsilon_V\mathds1_l\in dx\big\} \\
  &\qquad-\mfrac{\epsilon_V}{V}\#\big\{\text{jumps } \tilde x + \epsilon_V \mathds1_l\text{ to } \tilde x \text{ occurred in } (0,t): \tilde x + \tfrac{1}{2}\epsilon_V\mathds1_l\in dx\big\}.
\end{align*}
As usual $dx$ denotes a spatial area, possibly a small box surrounding one lattice site, and $\rho\super{V}_{t}$ and $\bW\super{V}_{t,l}$ are measures.

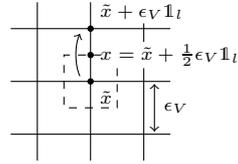
\begin{wrapfigure}{r}{0.4\linewidth}
\centering
\vspace{-0.3cm}
\begin{tikzpicture}[font=\scriptsize,scale=0.7]
  \draw(-1.5,-1)--(1.5,-1);
  \draw(-1.5, 0)--(1.5, 0);
  \draw(-1.5, 1)--(1.5, 1);
  \draw(-1,-1.5)--(-1,1.5);
  \draw(0, -1.5)--(0,0)--(0, 1.5);
  \draw(1, -1.5)--(1, 0.25); \draw(1,0.75)--(1,1.1);
  \draw[<->](1.2,-0.95)-- node[anchor=west,midway]{$\epsilon_V$}(1.2,-0.05);
  \draw[dashed](0.5,0.25)--(0.5,-0.5)--(-0.5,-0.5)--(-0.5,0.5)--(0.2,0.5);
  \filldraw(0,0) node[anchor=north west]{$\tilde x$} circle (0.05);
  \filldraw(0,1) node[anchor=south west]{$\tilde x+\epsilon_V\mathds1_l$} circle (0.05);
  \filldraw(0,1/2) node[anchor=west]{$x=\tilde x + \tfrac12\epsilon_V\mathds1_l$} circle (0.05);
  \draw[->] (-0.2,0.1).. controls (-0.3,0.3) and (-0.3,0.7)..(-0.2,0.9);
\end{tikzpicture}
\caption{A jump through midpoint $x$ in the direction with unit vector $\mathds1_l$.}
\label{fig:random walkers flux}
\end{wrapfigure}
Now $\rho\super{V}_t(dx)$ measures the number of particles present in (lattice points in) an area $dx$, while $\bW\super{V}_{t,l}(dx)$ measures the \emph{net} number of particles that have jumped through all midpoints in $dx$, in direction $\mathds1_l$, for $l=1,\hdots,d$, see Figure~\ref{fig:random walkers flux}. Note that both $\rho\super{V}$ and $\bW\super{V}$ are defined as measures on the lattice with shrinking distance $\epsilon_V$ between lattice points; this measure-valued formulation is needed to pass to a continuum limit later on.

The concentrations and fluxes are related by the $V$-dependent continuity equation:
\begin{align}
  \rho\super{V}_t(dx) &= \phi\super{V}\lbrack\bW\super{V}_t\rbrack(dx) := \big(\rho\super{V}_0 - \div\super{\epsilon_V}\bW\super{V}_t\big)(dx) \notag\\
    &:= \rho\super{V}_0(dx) - \mfrac1{\epsilon_V} \sum_{l=1}^d \big\lbrack \bW_{t,l}(dx + \tfrac12\epsilon+V\mathds1_l) - \bW_{t,l}(dx-\tfrac12\epsilon_V\mathds1_l) \big\rbrack.
    \label{eq:indepRW cont eq n}
\end{align} 

Using that $(\tfrac12\epsilon_V+\epsilon_V\Z)^d\subset\RR^d$, the integrated flux $\bW\super{V}_t$ is a Markov process in $\M(\RR^d)$ with generator
\begin{align}
  &(\Q\super{V}f)(\bar w) :=    \label{eq:indepRW generator}\\
  &\quad \mfrac{V}{\epsilon_V^2}\!\int\!\phi\super{V}\lbrack\bar w\rbrack(dx) \sum_{l=1}^d \big( f(\bar w-\tfrac{\epsilon_V}{V}\delta_{x-(\epsilon_V/2)\mathds1_l}) - 2f(\bar w)   
+ f(\bar w+\tfrac{\epsilon_V}{V}\delta_{x+(\epsilon_V/2)\mathds1_l}) 
    \big). \notag
\end{align}
Here, the factor $V/\epsilon_V^2$ comes from the time scaling $\epsilon_V^2$, together with the fact that we have $V\rho(dx)$ independent particles to choose from.

\subsection{Limit and large deviations}
\label{subsec:indep limit ldp}

For a test function $f \in C^1_b\big(\M(\RR^d)\big)$, we set:
\begin{align*}
  df(\bar w)_l(x) &:=\lim_{\tau\to0} \frac{ f(\bar w+\tau\delta_x e_l) - f(\bar w) }{\tau}.
\end{align*}
The continuity equation~\eqref{eq:indepRW cont eq n} converges to the limit continuity equation (with the usual divergence operator):
\begin{align}
  \rho_t(dx) = \phi\lbrack\bar w_t\rbrack(dx) &:= \rho_0(dx) - \div_x\bar w_t(dx).
    \label{eq:indepRW cont eq}
\end{align} 
With this notation, as $V\to\infty$ the generator \eqref{eq:indepRW generator} converges to (if $\phi\lbrack \bar w\rbrack(dx)=\phi\lbrack\bar w\rbrack(x)\,dx$):
\begin{align*}
  (\Q\super{\infty}f)(\bar w) &:=\int\!\div_x df(\bar w)(x)\,\phi\lbrack\bar w\rbrack(dx)
    =
    -\int\!df(\bar w)(x)\cdot\grad_x\phi\lbrack\bar w\rbrack(x)\,dx. 
\end{align*}
As in Subsection~\ref{subsec:RRE1 dynamics limit ldp}, the limit generator depends on derivatives of the test function only, and so the process $\bar W\super{V}_t$ converges (pathwise in probability) to the deterministic path satisfying Fick's Law:
\begin{equation*}
  \dot{\bar w}_t = -\grad_x\phi\lbrack \bar w_t\rbrack.
\end{equation*}
Naturally, combining this equation with the continuity equation~\eqref{eq:indepRW cont eq} yields the diffusion equation for the empirical measure:
\begin{equation*}
  \dot\rho_t = \tfrac{d}{dt} \phi\lbrack\bar w_t\rbrack = d\phi_{\bar w_t}\dot{\bar w}_t = \lapl_x\phi\lbrack\bar w_t\rbrack = \lapl_x\rho_t.
\end{equation*}

Similarly, we derive the dynamic large deviations by studying the non-linear generator:
\begin{align}
  (\H\super{V}f)(\bar w)&:=\mfrac1V \e^{-V f(\bar w)} \big(\Q\super{V} \e^{Vf}\big)(\bar w) \notag\\
    &= \mfrac1{\epsilon_V^2} \int\!\phi\super{V}\lbrack\bar w\rbrack(dx)\times \notag\\
      &\hspace{0.5cm} \sum_{l=1}^d \e^{V f(\bar w-\tfrac{\epsilon_V}{V}\delta_{x-(\epsilon_V/2)\mathds1_l}) - Vf(\bar w)} -2 
                               +\e^{Vf(\bar w+\tfrac{\epsilon_V}{V}\delta_{x+(\epsilon_V/2}\mathds1_l) - Vf(w)} \notag\\
    &\xrightarrow{V\to\infty} \int\!\big( \div_x df(\bar w)(x) + \big\lvert df(\bar w)(x) \big\rvert^2\big) \, \phi\lbrack \bar w\rbrack(dx),
\label{eq:indepRW nonlinear semigroup}
\end{align}
which follows from expanding the exponentials (with order-$\epsilon_V$ exponents) up to second order. Then the following large-deviation principle on flux space holds:
\begin{align}
  &\hspace{1cm}\Prob\super{V}\!\big(\bW\super{V}_{(\cdot)}\approx {\bar w}_{(\cdot)}\big) \stackrel{V\to\infty}{\sim}  \e^{-V\int_0^T\!\L(\bar w_t,\dot{\bar w}_t)\,dt}, \notag
\intertext{with}
  \H(\bar w,\bar\zeta)&:=\int\!\big( \div_x\bar\zeta(x)+\lvert\bar\zeta(x)\rvert^2 \big) \phi\lbrack\bar{w}\rbrack(dx) 
    =\lVert\bar\zeta\rVert_{L^2(\phi\lbrack\bar w\rbrack)}^2 - \langle \bar\zeta,\grad_x \phi\lbrack\bar w\rbrack\rangle, \quad\text{and}   \label{eq:indepRW H}\\
  \L(\bar w,\bar\jmath)&:=\sup_{\bar\zeta} \,\langle\bar\zeta,\bar\jmath\rangle - \H(\bar w,\bar\zeta)
    =\mfrac14\lVert \bar\jmath + \grad_x\phi\lbrack\bar w\rbrack \rVert_{L^2(1/\phi\lbrack\bar w\rbrack)}^2.
  \notag
\end{align}

Note that $d\phi_{\bar w}=-\div_x$ is independent of $\bar w$, and by~\eqref{eq:indepRW H}, the invariance condition~\eqref{eq:gent invariance} is satisfied. Hence by the contraction principle~\cite[Th.~4.2.1]{Dembo1998}, one obtains the large-deviation principle corresponding to the states (empirical measures):
\begin{align}
  &\hspace{2cm}\Prob\super{V}\!\big(\rho\super{V}_{(\cdot)}\approx \rho_{(\cdot)}\big) \stackrel{V\to\infty}{\sim}  \e^{-V\int_0^T\!\hat\L(\rho_t,\dot\rho_t)\,dt}, \notag
\intertext{with}
  &\hat\L(\phi\lbrack \bar w\rbrack,s):=\inf_{\substack{\bar\jmath\in T_{\bar w}:\,-\!\div_x\bar\jmath = s}} \L(\bar w ,\bar\jmath)
    =\mfrac14\lVert s-\lapl_x\phi\lbrack\bar w\rbrack\rVert^2_{\mathring H^{-1}(\phi\lbrack\bar w\rbrack)},
  \label{eq:indepRW hatL}
\end{align}
using the notation $\lVert s\rVert_{\mathring H^{-1}(\rho)}^2:=\sup_{\xi} 2 s\cdot\xi - \lVert\xi\rVert_{\mathring H^{1}(\rho)}^2:=\sup_{\xi} 2 s\cdot\xi - \lVert\grad_x\xi\rVert_{L^2(\rho)}^2$.

\subsection{Induced GGSs in flux and state space}

We can now apply Theorem~\ref{th:gent MPR} to extract a GGS from the L-function $\L$. We first choose the `naive' flat manifold of non-negative vector measures $\W:=\M_+(\RR^d;\RR^d)$ (equipped with the flat total variation metric). It is easily checked that condition~\eqref{eq:gent MPR condition on H} holds for the free energy given by
\footnote{This expression can again be seen as a relative entropy, cf.~\eqref{eq:RRE2 F}, but now with respect to the Lebesgue measure, where the measure of the whole space -- in this case infinity -- is omitted. See also~\cite[Prop.~3.2]{MielkePeletierRenger2014} for a general result in locally finite measure spaces.}: 
\begin{equation}
  \F(\bar w):=\mfrac12\int\!\phi\lbrack\bar w\rbrack(dx)\log\phi\lbrack\bar w\rbrack(x) - \phi\lbrack\bar w\rbrack(dx),
\label{eq:indepRW F}
\end{equation}
where we identify $\phi\lbrack\bar w\rbrack(dx) = \phi\lbrack\bar w\rbrack(x)\,dx$, and we implicitly set $\F(w)=\infty$ whenever the measure is not absolutely continuous. The dissipation potentials are obtained from~\eqref{eq:MPR Psis from H} and \eqref{eq:gent Psis from Psi}, which yields
\begin{align}
  \Psi^*(\bar w,\bar\zeta):=\lVert\bar\zeta\rVert^2_{L^2(\phi\lbrack\bar w\rbrack)}
  &&\text{and}&&
  \Psi(\bar w,\bar\jmath):=\mfrac14\lVert \bar\jmath \rVert^2_{L^2(1/\phi\lbrack\bar w\rbrack)}.
\label{eq:indepRW Psis and Psi}
\end{align}
Theorem~\ref{th:gent MPR} states that $\L$ induces the GGS $(\M_+(\RR^d),\Psi,\F)$ on flux space. 

For the state space, it is well-known that the state L-function~\eqref{eq:indepRW hatL} induces the entropy-Wasserstein gradient flow of the entropy functional~\cite{Adams2011,Duong2013a,Erbar2015,MielkePeletierRenger2014}. By the theory developed in Section~\ref{sec:general theory}, we can now see how this gradient structure is related to the flux gradient structure. Indeed, the flux free energy $\F$ depends on state only, i.e. $\F(\bar w)=\hat\F(\phi\lbrack\bar w\rbrack))$, where $\hat\F(\rho)=\mfrac12\int\!\rho(dx)\log\rho(x) - \rho(dx)$, and so by Proposition~\ref{prop:gent flux pGGEN then state GGS} the state L-function $\hat\L$ induces a GGS driven by $\hat\F$. Moreover, since the invariance condition~\eqref{eq:gent invariance} holds, the dissipation potentials are related by \eqref{eq:gent hat diss from diss}:
\begin{align*}
  \hat\Psi^*(\rho,\xi) &= \lVert\grad_x\xi\rVert^2_{L^2(\rho)}=:\lVert\xi\rVert^2_{\mathring H^1(\rho)},\\
  \hat\Psi(\rho,s)&=\inf_{-\div_x\bar\jmath=s} \mfrac14\lVert \bar\jmath \rVert^2_{L^2(1/\rho)} =: \mfrac14\lVert s\rVert^2_{\mathring H^{-1}(\rho)}.
\end{align*}

\subsection{A new geometry}

The form of the dissipation potential $\Psi$ and $\hat\Psi$ suggests that it is more natural to use different manifolds in the spirit of Remark~\ref{rem:new manifold}. For the state space this points to the space $\W:=\P_2(\RR^d)=\{\rho\in\P(\RR^d):\int\!x^2\,\rho(dx)<\infty\}$ of probability measures of finite second moment space, equipped with the Monge-Kantorovich-Wasserstein metric:
\begin{equation*}
  d_{\P_2}(\rho_0,\rho_1)^2:=\inf_{\substack{\gamma\in \P(\RR^d\times\RR^d):\\ \gamma(\cdot\times\RR^d)=\rho_0(\cdot)\\ \gamma(\RR^d\times\cdot)=\rho_1(\cdot)}} \iint_{\RR^d\times\RR^d}\!\lvert x-y\rvert^2\,\gamma(dx\,dy).
\end{equation*}
with tangent and cotangent space $T_\rho=\mathring H^{-1}(\rho)$ and $T_\rho^*=\mathring H^{1}(\rho)$. For this setting, the inverse metric tensor $K_{\P_2}(\rho):T_\rho^*\to T_\rho$ is known by the Benamou-Brenier formula~\cite[Th.~8.1]{Villani2009} to be $K_{\P_2}(\rho)\xi:=-2\div_x\rho\grad_x\xi$, so that the GGS is indeed the entropy-Wasserstein gradient flow~\cite{Jordan1998}:
\begin{equation*}
  \lapl_x\rho_t  = \dot\rho_t = d_\xi\hat\Psi^*\big(\rho_t,-d\hat\F(\rho_t)\big)=-2K_{\P_2}(\rho_t)d\hat\F(\rho_t) =: -2\Grad_{\rho_t}\hat\F(\rho_t).
\end{equation*}

Motivated by this observation we can take for the flux manifold the space of signed vector measures of finite \emph{first} moment $\M_1(\RR^d;\RR^d):=\{\bar w\in \M(\RR^d;\RR^d): \int\!\lvert x\rvert\,\lvert \bar w\rvert(dx)<\infty\}$. This choice guarantees that the corresponding states have finite second moment (once $\rho_0\in\P_2(\RR^d)$):
\begin{equation*}
  \int\!\lvert x\rvert^2\,\rho(dx) = \int\!\lvert x\rvert^2\,\rho_0(dx) + \int\!x\cdot \bar w(dx) - \textit{non-neg. bnd. term.} <\infty. 
\end{equation*}
Moreover, we can now use the dissipation potential to construct a natural metric on $\Y$:
\begin{equation}
  d_{\M_1}(\bar w_0,\bar w_1)^2:=\inf_{\substack{\tilde w:(0,1)\to\W:\\\tilde w_0=\bar w_0,\,\tilde w_1=\bar w_1}} \int_0^1\!\lVert\dot {\tilde w}_t\rVert^2_{L^2(1/(\rho_0-\div_x\tilde w_t))}\,dt,
\label{eq:indepRW metric}
\end{equation}
where the infimum runs over paths of fluxes for which $\rho_0-\div\tilde w_t$ remains non-negative. The corresponding tangent and cotangent spaces (in the interior of the domain) are simply $T_{\bar w} = L^2(1/(\rho_0-\div_x\bar w))$ and $T_{\bar w}^*=L^2(\rho_0-\div_x\bar w)$ and the inverse metric tensor $K_{\M_1}(\bar w):T_{\bar w}^*\to T_{\bar w}$ is $K_{\M_1}(\bar w)\bar\zeta=\bar\zeta/(\rho_0-\div_x\bar w)$. This yields an interesting geometry in flux space, which, as far as the author is aware, is still unknown in the literature.


\section{A simple reaction-diffusion model}
\label{sec:reaction-diffusion}

We now combine the models from Section~\ref{sec:reaction fluxes} and \ref{sec:diffusion} to study reaction-diffusion models in flux and state space. The stochastic particle system will now consist of `reacting random walkers'. It is known that, if the reaction networks include reactions of different orders (unimolecular, bimolecular, etc.) and we only allow particles to react if the required number of particles are present within the same site/compartment, then the model may not converge to the expected reaction-diffusion equation~\cite{PfaffelhuberPopovic}. The reason behind this is that for a multimolecular reaction, it becomes very unlikely that the required amount of reactants are all within one site/compartment; different order reactions would require different scalings. This is beyond the scope of the current paper. However we can already illustrate the combination of reaction and transport fluxes for a simple system of unimolecular equations of the type:
\begin{align*}
  \sfA \xrightarrow{\kappa_\fw} \sfB &&\text{and}&& \sfB\xrightarrow{\kappa_\bw} \sfA.
\end{align*}
In this section we consider GGSs only, hence we shall always consider net rather than one-way fluxes.

\subsection{Reacting and diffusing particle system}

Since we consider unimolecular reactions only, we can take \emph{independent} reacting random walkers on the scaled lattice $(\epsilon_V\Z)^d$, where each reaction occurs locally at each lattice site with rate $\kappa_\fw$ or $\kappa_\bw$ respectively, so that
\begin{align*}
  \tfrac1V\lambda\super{V}_\fw(\rho(dx))\equiv\kappa_\fw\rho_\sfA(dx)
  &&\text{and}&&
  \tfrac1V\lambda\super{V}_\bw(\rho(dx))\equiv\kappa_\bw\rho_\sfB(dx).
\end{align*}
For the transport mechanism, we assume that the two species $\sfA,\sfB$ hop to neighbouring lattice sites with rates $D_\sfA$ and $D_\sfB$ respectively.

As before we consider the random concentrations, as well as the integrated (net) fluxes, where we now distinguish between \underline{tr}ansport fluxes and \underline{re}action fluxes. If $X_{t,i}\in(\epsilon_V\Z^d)\subset\RR^d$ is the position and $Y_{t,i}\in\Y=\{\sfA,\sfB\}$ is the species of the $i$-th particle, then
\begin{align*}
  &\rho\super{V}_{t,y}(dx):=\mfrac1V\#\{i=1,\hdots,V: X_{t,i}\in dx \text{ and } Y_{t,i}=y\},\\
  &\bW\super{V}_{t,\tr,y,l}(dx):=\\
  &\quad \mfrac{\epsilon_V}{V}\#\big\{\text{jumps } \tilde x \text{ to } \tilde x + \epsilon_V \mathds1_l \text{ of species $y$ occurred in } (0,t): \tilde x + \tfrac{1}{2}\epsilon_V\mathds1_l\in dx\big\} \\
  &\quad-\mfrac{\epsilon_V}{V}\#\big\{\text{jumps } \tilde x + \epsilon_V \mathds1_l\text{ to } \tilde x \text{ of species $y$ occurred in } (0,t): \tilde x + \tfrac{1}{2}\epsilon_V\mathds1_l\in dx\big\}, \\
  &\bW\super{V}_{t,\re}(dx) := \mfrac1V\#\big\{\text{forward reactions } r \text{ occurred in } (0,t) \text{ and in area } dx \big\},\\
  &\qquad\qquad\quad - \mfrac1V\#\big\{\text{backward reactions } r \text{ occurred in } (0,t) \text{ and in area } dx\big\}.
\end{align*}
The concentrations and fluxes are again related by a continuity equation:
\begin{align}
  \rho\super{V}_{t,y}(dx) &= \phi\super{V}\lbrack\bW\super{V}_t\rbrack_y(dx) := \big(\rho\super{V}_{0} - \div\super{\epsilon_V}\bW\super{V}_{t,\tr} + \Gamma \bW\super{V}_{t,\re}\big)_y(dx),
\label{eq:rediff V-cont eq}
\end{align} 
where the discrete divergence is as in \eqref{eq:indepRW cont eq n}, and $\Gamma=(-1,1)$ is the matrix consisting of one state change vector corresponding to a forward reaction.

The pair $\bW\super{V}_t:=(\bW\super{V}_{t,\tr},\bW\super{V}_{t,\re})$ is then a Markov process with generator
\begin{equation*}
  (\Q\super{V}f)(\bar w_\tr,\bar w_\re) := (\Q_\tr\super{V}f)(\bar w_\tr,\bar w_\re) + (\Q\super{V}_\re f)(\bar w_\tr,\bar w_\re),
\end{equation*}
where
\begin{align*}
  &(\Q_\tr\super{V}f)(\bar w_\tr,\bar w_\re) \\
  &\hspace{1.4cm}:= \mfrac{V}{\epsilon_V^2} \sum_{y=\sfA,\sfB}D_y\!\int\!\phi\super{V}\lbrack\bar w\rbrack_y(dx) \sum_{l=1}^d
  \big( f(\bar w_\tr - \tfrac{\epsilon_V}{V}\delta_{x-(\epsilon_V/2)\mathds1_l}\mathds1_y,\bar w_\re) \\
  &\hspace{4cm}- 2f(\bar w_\tr,\bar w_\re)
               + f(\bar w_\tr+\tfrac{\epsilon_V}{V}\delta_{x+(\epsilon_V/2)\mathds1_l}\mathds1_y,\bar w_\re) \big), \\
  &(\Q_\tr\super{V}f)(\bar w_\tr,\bar w_\re)\\
  &\hspace{1.4cm}:=V \kappa_\fw \int\!\phi\super{V}\lbrack\bar w\rbrack_\sfA(dx) \big( f(\bar w_\tr,\bar w_\re +\tfrac1V\delta_x) - f(\bar w_\tr,\bar w_\re)\big)\\
  &\hspace{1.4cm}\quad+ V \kappa_\bw \int\!\phi\super{V}\lbrack\bar w\rbrack_\sfB(dx) \big( f(\bar w_\tr,\bar w_\re -\tfrac1V\delta_x) - f(\bar w_\tr,\bar w_\re)\big).
\end{align*}

\subsection{Limit and large deviations}

By the same procedure as in Sections~\ref{subsec:RRE1 dynamics limit ldp} and \ref{subsec:indep limit ldp}, one finds that as $V\to\infty$ and $\epsilon_V\to0$, the continuity operator~\eqref{eq:rediff V-cont eq} converges to (assuming $\phi\lbrack \bar w\rbrack_y(dx) = \phi\lbrack \bar w\rbrack_y(x)\,dx$):
\begin{equation}
  \phi\lbrack \bar w\rbrack_y(x):=\rho_{0,y}(x) - \div_x \bar w_{\tr,y}(x) + (\Gamma\bar w_{\re})_y(x),
\label{eq:rediff cont eq}
\end{equation}
and the process converges (pathwise in probability) to the solution of the system:
\begin{align*}
  \dot{\bar w}_{t,\tr,y}(x) &= -D_y\grad_x\phi\lbrack \bar w_t\rbrack_y(x),\\
    \dot{\bar w}_{t,\re}(x) &= \kappa_\fw \phi\lbrack \bar w_t\rbrack_\sfA(x) - \kappa_\bw \phi\lbrack \bar w_t\rbrack_\sfB(x).
\end{align*}
Indeed, putting these together yields the reaction-diffusion equation for the limit concentrations:
\begin{align*}
  \dot\rho_{t,\sfA}(x)= D_\sfA\lapl_x \rho_{t,\sfA}(x) - \kappa_\fw\rho_{t,\sfA}(x) + \kappa_\bw\rho_{t,\sfB}(x),\\
  \dot\rho_{t,\sfB}(x)= D_\sfB\lapl_x \rho_{t,\sfB}(x) + \kappa_\fw\rho_{t,\sfA}(x) - \kappa_\bw\rho_{t,\sfB}(x).
\end{align*}

To find the corresponding large deviations, we combine \eqref{eq:RRE2 nonlinear semigroup} and \eqref{eq:indepRW nonlinear semigroup} to calculate the non-linear generator:
\begin{multline*}
  (\H\super{V}f)(\bar w_\tr,\bar w_\re) :=\mfrac1V \e^{-V f(\bar w_\tr,\bar w_\re)} \big(\Q\super{V} \e^{Vf}\big)(\bar w_\tr,\bar w_\re)\\
  \xrightarrow{V\to\infty} \sum_{y=\sfA,\sfB}D_y\int\!\big( \div_x \partial_{\bar w_{\tr,y}} f(\bar w_\tr,\bar w_\re)(x) + \big\lvert \partial_{\bar w_{\tr,y}} f(\bar w_\tr,\bar w_\re)(x) \big\rvert^2\big) \, \phi\lbrack \bar w\rbrack_y(dx)\\
   + \kappa_\fw\int\!\phi\lbrack\bar w \rbrack_\sfA(dx) \big( \e^{\partial_{\bar w_\re}f(\bar w_\tr,\bar w_\re)} -1 \big)\\
   + \kappa_\bw\int\!\phi\lbrack\bar w \rbrack_\sfB(dx) \big( \e^{-\partial_{\bar w_\re}f(\bar w_\tr,\bar w_\re)} -1 \big).
\end{multline*}

Let us again abbreviate $\bar\zeta=(\bar\zeta_{\tr,y,l}(x),\bar\zeta_{\re}(x))$ and $\bar\jmath = (\bar\jmath_{\tr,y,l}(x),\bar\jmath_{\re}(x))$. The limiting non-linear generator can now be split into
\begin{align}
  &\H(\bar w,\bar\zeta) :=\H_\tr(\bar w,\bar\zeta_\tr)+\H_\re(\bar w,\bar\zeta_\re),  \label{eq:rediff H} \\
  &\quad\H_\tr(\bar w,\bar\zeta_\tr) :=\sum_{y=\sfA,\sfB} D_y\big(\lVert\bar\zeta_{\tr,y}\rVert_{L^2(\phi\lbrack\bar w\rbrack_y)}^2 - \langle \bar\zeta_{\tr,y},\grad_x \phi\lbrack\bar w\rbrack_y\rangle\big), \notag\\
  &\quad\H_\re(\bar w,\bar\zeta_\re) :=\kappa_\fw\int\!\phi\lbrack\bar w\rbrack_\sfA(dx) \big( \e^{\bar\zeta_\re(x)}-1\big) + \kappa_\bw\int\!\phi\lbrack\bar w\rbrack_\sfB(dx) \big( \e^{-\bar\zeta_\re(x)}-1\big). \notag
\end{align}
Since each mechanism corresponds to a separate flux, the corresponding L-function also splits into two parts:
\begin{align}
  &\L(\bar w,\bar\jmath):=\sup_{\bar\zeta_\tr,\bar\zeta_\re} \,\langle\bar\zeta_\tr,\bar\jmath_\tr\rangle + \langle\bar\zeta_\re,\bar\jmath_\re\rangle- \H(\bar w,\bar\zeta)
  :=\L_\tr(\bar w,\bar\jmath_\tr)+\L_\re(\bar w,\bar\jmath_\re),
  \label{eq:rediff L}\\
  &\quad\L_\tr(\bar w,\bar\jmath_\tr)
    :=\sum_{y=\sfA,\sfB}\mfrac1{4D_y}\lVert \bar\jmath_{\tr,y} + D_y\grad_x\phi\lbrack\bar w\rbrack_y \rVert_{L^2(1/\phi\lbrack\bar w\rbrack_y)}^2, \notag\\
  &\quad\L_\re(\bar w,\bar\jmath_\re)
    :=\inf_{j_\fw-j_\bw=\bar\jmath_\re} h\big(j_\fw | \kappa_\fw\phi\lbrack \bar w\rbrack_\sfA\big) + h\big(j_\bw | \kappa_\bw\phi\lbrack\bar w\rbrack_\sfB\big), \notag
\end{align}
using the usual the relative entropy between two measures, i.e: $h(j| k):=\int\! j(dx)\log(dj/dk(x)) - j(dx) + k(dx)$ if $j\ll k$, else $h(j| k):=\infty$.

As before, the calculation above formally shows that the flux large-deviation principle holds (see for example~\cite{BodineauLagouge2012} for a similar but rigorous result):
\begin{equation*}
  \Prob\super{V}\!\big(\bW\super{V}_{(\cdot)}\approx {\bar w}_{(\cdot)}\big) \stackrel{V\to\infty}{\sim}  \e^{-V\int_0^T\!\L(\bar w_t,\dot{\bar w}_t)\,dt}.
\end{equation*}

The L-function~\eqref{eq:rediff L} splits into two parts because the only interaction between the two mechanisms occurs through the state $\phi\lbrack\bar w\rbrack$. By contrast, the corresponding state space large-deviation are much more complicated. Observe that the continuity equation~\eqref{eq:rediff cont eq} is an affine function of $\bar w$, and so $d\phi_{\bar w}\bar j$ is independent of $\bar w$, and the invariance condition~\eqref{eq:gent invariance} holds. As explained in the beginning of Section~\ref{sec:general theory}, this means that one can apply a straightforward contraction principle on the tangents to yield the large deviation cost function for the states/concentrations:
\begin{align}
  \hat \L(\phi\lbrack\bar w\rbrack,s)&:=\inf_{\substack{\bar\jmath=(\bar\jmath_\tr,\bar\jmath_\re):\\s=-\div\bar\jmath_\tr+\Gamma\bar\jmath_\re}} \L(\bar w,\bar\jmath).
\label{eq:rediff hat L}
\end{align}
This infimum reintroduces a strong interrelation between the two driving mechanisms. Indeed, for a given tangent $(\rho,s)$, the fluxes in this infimum correspond to an optimal splitting between the two mechanisms, which can be seen as an inf-convolution. Similar interactions also arise when considering multiple reaction pairs, see \cite[Sect.~3.4]{MielkePattersonPeletierRenger2017}.

\subsection{GGSs in flux and state space}

We now apply Theorem~\ref{th:gent MPR} to the reaction-diffusion setting. The symmetry condition~\eqref{eq:gent MPR condition on H} holds for the function~\eqref{eq:rediff H} if we choose the free energy functional
\begin{align*}
  \F(\bar w) :=& \mfrac12\int\!\phi\lbrack\bar w\rbrack_\sfA(dx)\log \kappa_\fw\phi\lbrack\bar w\rbrack_\sfA(x) -\phi\lbrack\bar w\rbrack_\sfA(dx)\\
               &+\mfrac12\int\!\phi\lbrack\bar w\rbrack_\sfB(dx)\log \kappa_\bw\phi\lbrack\bar w\rbrack_\sfB(x) -\phi\lbrack\bar w\rbrack_\sfB(dx).
\end{align*}
Naturally, this functional can be seen as a combination of~\eqref{eq:RRE2 F} and \eqref{eq:indepRW F}, where just like~\eqref{eq:indepRW F}, it has the form of a relative entropy with respect to a locally finite invariant measure, namely $(\pi_y(dx))_{y=\sfA,\sfB}\equiv (1/\kappa_\fw,1/\kappa_\bw)\,dx$. We find the corresponding dissipation potentials from~\eqref{eq:gent MPR GGEN Psis from H}, which is again a combination of the non-quadratic potentials~\eqref{eq:RRE2 Psis},\eqref{eq:RRE2 Psi} and the quadratic potential~\eqref{eq:indepRW Psis and Psi}:
\begin{align*}
  &\Psi^*(\bar w,\bar\zeta) :=  \Psi^*_\tr(\bar w,\bar\zeta_\tr) +  \Psi^*_\re(\bar w,\bar\zeta_\re), \\
  &\quad\Psi^*_\tr(\bar w,\bar\zeta_\tr):=D_\sfA \lVert\bar\zeta_{\tr,\sfA}\rVert_{L^2(\phi\lbrack\bar w\rbrack_\sfA)}^2 + D_\sfB \lVert\bar\zeta_{\tr,\sfB}\rVert_{L^2(\phi\lbrack\bar w\rbrack_\sfB)}^2,\\
  &\quad\Psi^*_\re(\bar w,\bar\zeta_\re):= \int\!\sigma(\bar w)(x)\big(\cosh(\bar\zeta_{\re}(x))-1\big)\,dx, \qquad\text{ and}\\
  &\Psi(\bar w,\bar\jmath) := \Psi_\tr(\bar w,\bar\jmath_\tr) + \Psi_\re(\bar w,\bar\jmath_\re), \\
  &\quad\Psi_\tr(\bar w,\bar\jmath_\tr):=\mfrac1{4D_\sfA} \lVert\bar\jmath_{\tr,\sfA}\rVert_{L^2(1/\phi\lbrack\bar w\rbrack_\sfA)}^2 + \mfrac1{4D_\sfB} \lVert\bar\jmath_{\tr,\sfB}\rVert_{L^2(1/\phi\lbrack\bar w\rbrack_\sfB)}^2,\\
  &\quad\Psi_\re(\bar w,\bar\jmath_\re):= \int\!\sigma(\bar w)(x)\big(\cosh^*(\mfrac{\bar\jmath_{\re}(x)}{\sigma(\bar w)(x)})+1\big)\,dx,
\end{align*}
with $\sigma(\bar w)(x):=2\sqrt{\kappa_\fw\kappa_\bw\phi\lbrack\bar w\rbrack_\sfA(x)\phi\lbrack\bar w\rbrack_\sfB(x)}$. Let the flux space be given by $\W=\M_1(\RR^d;\RR^d)\times\M_1(\RR^d;\RR^d)\times L^1(\RR^d)$, where the first two spaces, corresponding to the transport fluxes, are equipped with the metric~\eqref{eq:indepRW metric} introduced in the previous section. By Theorem~\ref{th:gent MPR} the flux cost function $\L$ induces the GGS $(\W,\Psi,\F)$. We stress that the dissipation potential $\Psi$ splits into two potentials for the transport and reaction mechanisms respectively, but the free energy is one and the same for both mechanisms.

Since a GGS is a special case of a pGGEN, by Theorem~\ref{th:gent flux pGGEN equiv state GGS}, the state cost function $\hat\L$ also induces a GGS $(\X,\hat\Psi,\hat\F)$, in this case in the space $\X=\P_2(\RR^d\times\{\sfA,\sfB\})$. The same result yields
\begin{equation*}
  \hat\F(\rho) := \mfrac12\int\!\rho_\sfA(dx)\log \kappa_\fw\rho_\sfA(x) - \rho_\sfA(dx) +\mfrac12\int\!\rho_\sfB(dx)\log \kappa_\bw\rho_\sfB(x) - \rho_\sfB(dx),
\end{equation*}
and, using $d\phi_{\bar w}\tp=\begin{pmatrix} \grad_x & 0\\ 0 & \grad_x\\ -1 & 1 \end{pmatrix}$:
\begin{align*}
  \hat\Psi^*(\rho,\xi)\,\, &\!\!\!\stackrel{\eqref{eq:gent hat Psis from Psis}}{=}\Psi^*(\bar w,d\phi\tp_{\bar w}\xi) = \Psi^*_\tr(\bar w, \grad_x\xi) + \Psi^*_\re(\bar w,\xi_\sfB-\xi_\sfA)\\
    &=  D_\sfA \lVert\grad_x\xi_\sfA\rVert_{L^2(\rho_\sfA)}^2 + D_\sfB \lVert\grad_x\xi_\sfB\rVert_{L^2(\rho_\sfB)}^2\\
      &\qquad + 2\int\!\sqrt{\kappa_\fw\kappa_\bw\rho_\sfA(x)\rho_\sfB(x)}\big(\cosh(\xi_\sfB(x)-\xi_\sfA(x))-1\big)\,dx,\\
  \hat\Psi(\rho,s)\,\,\,\,\,\,\,&\!\!\!\!\!\!\!\!\stackrel{\eqref{eq:gent hat Psi from Psi},\eqref{eq:rediff cont eq}}{=} \inf_{s = -\div_x \bar\jmath_{\tr} + \Gamma\bar\jmath_{\re}} \Psi_\tr(\bar w,\bar\jmath_\tr) + \Psi_\re(\bar w,\bar\jmath_\re).
\end{align*}
We stress that, analogous to the L-function~\eqref{eq:rediff hat L}, the dissipation potential $\hat\Psi$ on state space no longer splits into two parts.

\section{Discussion}

We studied gradient and (pre-) Generic structures induced by flux large deviations, and the relationship between structures induced by state large deviations. At a first glance, the resulting gradient or generic structures in flux may appear non-physical. However, in practice many induced flux structures have a free energy and dissipation potential that only depends on the integrated flux through the state of the system. Hence the main difference with energy-driven structures in state space is that the fluxes rather than velocities are being driven, which seems a very physical assumption.

It turns out that if the fluxes follow a gradient or generic flow, then so do the corresponding states. The same principle also holds in the other direction, but the flow in flux space could have an additional Hamiltonian term that is not observed when considering states only.

On a mathematical level, it can also be worthwhile to work with fluxes rather than states. As we saw in the examples, and Section~\ref{sec:reaction-diffusion} in particular, if a microscopic system consists of multiple driving mechanisms, then the corresponding Markov generator as well as the non-linear generator is a sum over these mechanisms. By considering separate fluxes for each of these mechanisms, the large-deviation L-function and its induced dissipation potential also splits into different terms for each mechanism. This decomposition can be beneficial for analytic and numerical purposes, e.g. using operator splitting techniques.




\section*{Acknowledgements}

This research has been funded by Deutsche Forschungsgemeinschaft (DFG) through grant CRC 1114 ``Scaling Cascades in Complex Systems'', Project C08 ``Stochastic spatial coagulation particle processes''. The author thanks H.C. \"Ottinger and his group for their valuable discussion and comments.

\bibliographystyle{alpha}
\bibliography{library}

\newcommand{\etalchar}[1]{$^{#1}$}
\begin{thebibliography}{BDSG{\etalchar{+}}15}

\bibitem[ADE17]{AgazziDemboEckmann2017a}
A.~Agazzi, A.~Dembo, and J.-P. Eckmann.
\newblock Large deviations theory for markov jump models of chemical reaction
  networks.
\newblock \href{https://arxiv.org/abs/1701.02126}{arxiv.org/abs/1701.02126},
  2017.

\bibitem[ADPZ11]{Adams2011}
S.~Adams, N.~Dirr, M.~A. Peletier, and J.~Zimmer.
\newblock From a large-deviations principle to the {W}asserstein gradient flow:
  a new micro-macro passage.
\newblock {\em Communications in Mathematical Physics}, 307(3):791--815, 2011.

\bibitem[AK11]{AndersonKurtz2011}
D.F. Anderson and T.G. Kurtz.
\newblock Continuous time {M}arkov chain models for chemical reaction networks.
\newblock In Heinz Koeppl, Gianluca Setti, Mario di~Bernardo, and Douglas
  Densmore, editors, {\em Design and Analysis of Biomolecular Circuits}, pages
  3--42. Springer, New York, N.Y. U.S.A, 2011.

\bibitem[BDSG{\etalchar{+}}04]{Bertini2004}
L.~Bertini, A.~De~Sole, D.~Gabrielli, G.~Jona-Lasinio, and C.~Landim.
\newblock Minimum dissipation principle in stationary non-equilibrium states.
\newblock {\em Journal of Statistical Physics}, 116(1):831--841, 2004.

\bibitem[BDSG{\etalchar{+}}15]{Bertini2015MFT}
L.~Bertini, A.~De~Sole, D.~Gabrielli, G.~Jona-Lasinio, and C.~Landim.
\newblock Macroscopic fluctuation theory.
\newblock {\em Reviews of Modern Physics}, 87(2), 2015.

\bibitem[BL12]{BodineauLagouge2012}
T.~Bodineau and M.~Lagouge.
\newblock Large deviations of the empirical currents for a boundary-driven
  reaction diffusion model.
\newblock {\em Annals of Applied Probability}, 22(6):2282--2319, 2012.

\bibitem[Die15]{Dietert2015}
H.~Dietert.
\newblock Characterisation of gradient flows on finite state {M}arkov chains.
\newblock {\em Electronic Communications in Probability}, 20(29):1--8, 2015.

\bibitem[DLR13]{Duong2013a}
M.~H. Duong, V.~Laschos, and M.~Renger.
\newblock {W}asserstein gradient flows from large deviations of many-particle
  limits.
\newblock {\em ESAIM: Control, Optimisation and Calculus of Variations},
  19(4):1166--1188, 2013.

\bibitem[DRW16]{DupuisRamananWu2016}
P.~Dupuis, K.~Ramanan, and W.~Wu.
\newblock Large deviation principle for finite-state mean field interacting
  particle systems.
\newblock \href{https://arxiv.org/abs/1601.06219}{arxiv.org/abs/1601.06219},
  2016.

\bibitem[DZ87]{Dembo1998}
A.~Dembo and O.~Zeitouni.
\newblock {\em Large deviations techniques and applications}, volume~38 of {\em
  Stochastic modelling and applied probability}.
\newblock Springer, New York, NY, USA, 2nd edition, 1987.

\bibitem[EMR15]{Erbar2015}
M.~Erbar, J.~Maas, and D.R.M. Renger.
\newblock From large deviations to {W}asserstein gradient flows in multiple
  dimensions.
\newblock {\em Electronic Communications in Probability}, 20(89):1--12, 2015.

\bibitem[FK06]{Feng2006}
J.~Feng and T.G. Kurtz.
\newblock {\em Large deviations for stochastic processes}, volume 131.
\newblock American Mathematical Society, Providence, RI, USA, 2006.

\bibitem[JKO98]{Jordan1998}
R.~Jordan, D.~Kinderlehrer, and F.~Otto.
\newblock The variational formulation of the {F}okker-{P}lanck equation.
\newblock {\em SIAM Journal on Mathematical Analysis}, 29(1):1--17, 1998.

\bibitem[KJZ17]{KaiserJackZimmer2017}
M.~Kaiser, R.L. Jack, and J.~Zimmer.
\newblock Canonical structure and orthogonality of forces and currents in
  irreversible markov chains.
\newblock \href{https://arxiv.org/abs/1708.01453}{arxiv.org/abs/1708.01453},
  2017.

\bibitem[KLMP18a]{KLMP2017phys}
R.~Kraaij, A.~Lazarescu, C.~Maes, and M.A. Peletier.
\newblock Deriving {GENERIC} from a generalized fluctuation symmetry.
\newblock {\em Journal of Statistical Physics}, 170(3):492--508, 2018.

\bibitem[KLMP18b]{KLMP2018math}
R.~Kraaij, A.~Lazarescu, C.~Maes, and M.A. Peletier.
\newblock Fluctuation symmetry leads to {GENERIC} equations with non-quadratic
  dissipation.
\newblock
  \href{https://arxiv.org/abs/1712.10217v1}{arxiv.org/abs/1712.10217v1}, 2018.

\bibitem[Kur72]{Kurtz1972}
T.~G. Kurtz.
\newblock The relationship between stochastic and deterministic models for
  chemical reactions.
\newblock {\em The Journal of Chemical Physics}, 57(7):2976--2978, 1972.

\bibitem[Mae17]{Maes2017}
C.~Maes.
\newblock Frenetic bounds on the entropy production.
\newblock {\em ArXiv Preprint}, 1705.07412, 2017.

\bibitem[Mie11]{Mielke2011GGEN}
A.~Mielke.
\newblock Formulation of thermoelastic dissipative material behavior using
  {GENERIC}.
\newblock {\em Continuum Mechanics and Thermodynamics}, 23(3):233--256, 2011.

\bibitem[MN08]{MaesNetocny2008}
C.~Maes and K.~Neto{\v{c}}n{\'y}.
\newblock Canonical structure of dynamical fluctuations in mesoscopic
  nonequilibrium steady states.
\newblock {\em EPL}, 82(3):30003, 2008.

\bibitem[MPPR17]{MielkePattersonPeletierRenger2017}
A.~Mielke, I.A. Patterson, M.A. Peletier, and D.R.M. Renger.
\newblock Non-equilibrium thermodynamical principles for chemical reactions
  with mass-action kinetics.
\newblock {\em {SIAM} {J}ournal on {A}pplied {M}athematics)}, 77(4), 2017.

\bibitem[MPR14]{MielkePeletierRenger2014}
A.~Mielke, M.A. Peletier, and D.R.M. Renger.
\newblock On the relation between gradient flows and the large-deviation
  principle, with applications to {M}arkov chains and diffusion.
\newblock {\em Potential Analysis}, 41(4), 2014.

\bibitem[OM53]{Onsager1953I}
L.~Onsager and S.~Machlup.
\newblock Fluctuations and irreversible processes.
\newblock {\em Phys. Rev.}, 91(6):1505--1512, Sep 1953.

\bibitem[Ons31]{Onsager1931I}
L.~Onsager.
\newblock Reciprocal relations in irreversible processes {I}.
\newblock {\em Phys. Rev.}, 37(4):405--426, Feb 1931.

\bibitem[{\"O}tt05]{Ottinger2005}
H.C. {\"O}ttinger.
\newblock {\em Beyond equilibrium thermodynamics}.
\newblock Wiley-Interscience, Hoboken, NJ, USA, 2005.

\bibitem[PP15]{PfaffelhuberPopovic}
P.~Pfaffelhuber and L.~Popovic.
\newblock Scaling limits of spatial compartment models for chemical reaction
  networks.
\newblock {\em The Annals of Applied Probability}, 25(6):3162--3208, 2015.

\bibitem[PR18]{PattersonRenger2018flux}
R.I.A. Patterson and D.R.M. Renger.
\newblock Large deviations of reaction fluxes.
\newblock \href{https://arxiv.org/abs/1802.02512}{arxiv.org/abs/1802.02512},
  2018.

\bibitem[Ren17]{Renger2017}
D.R.M. Renger.
\newblock Flux large deviations of independent and reacting particle systems,
  with implications for macroscopic fluctuation theory.
\newblock WIAS Preprint No. 2375, 2017.

\bibitem[Vil09]{Villani2009}
C.~Villani.
\newblock {\em Optimal transport: old and new}.
\newblock Springer, Berlin-Heidelberg, Germany, 2009.

\end{thebibliography}

\end{document}